\documentclass[12pt,letterpaper]{article}
\usepackage{amsfonts}
\usepackage{amsmath}
\usepackage{amsthm}
\usepackage{amssymb}
\usepackage{geometry}
\usepackage{xcolor}
\usepackage{graphicx}
\usepackage[onehalfspacing]{setspace}
\usepackage[authoryear,round,longnamesfirst]{natbib}
\usepackage{grffile}
\usepackage{amssymb}
\usepackage{amsmath}
\usepackage{amsmath}
\usepackage{amsthm}
\usepackage{amsfonts}
\usepackage{graphicx}
\usepackage{enumerate}
\usepackage{mathrsfs}  
\usepackage{bm}
\usepackage{todonotes}
\usepackage{graphicx}
\usepackage{dsfont}
\usepackage[authoryear,longnamesfirst]{natbib}
\usepackage{fullpage}
\usepackage{subcaption}
\usepackage{mathabx}
 \usepackage{xr}
\usepackage{multirow}
\usepackage{comment}
\usepackage{multicol}
\usepackage{tikz}
\numberwithin{equation}{section}
\newtheorem{theorem}{Theorem}
\newtheorem{algorithm}{Algorithm}

\newtheorem{proposition}{Proposition}

\newtheorem{lemma}{Lemma}
\usepackage{epigraph}
\theoremstyle{definition}
\newtheorem{assumption}{Assumption}

\newtheorem{remark}{Remark}

\renewcommand{\Pr}{\text{Pr}}

\renewcommand{\tilde}{\widetilde}

\newcommand{\Cov}[0]{\mathrm{Cov}}
\newcommand{\Var}[0]{\mathrm{Var}}

\newcommand{\E}[0]{\mathbb{E}}

\newcommand{\sumg}{\sum_{g=1}^G}
\newcommand{\sumi}{\sum_{i=1}^{N_g} }
\newcommand{\1}{\mathds{1}}
\newcommand{\Real}{\mathbb{R}}

\DeclareMathOperator{\argmin}{argmin}

\renewcommand{\hat}{\widehat}
\renewcommand{\tilde}{\widetilde}
\renewcommand{\bar}{\overline}
\renewcommand{\check}{\widecheck}

\allowdisplaybreaks

\begin{document}
\allowdisplaybreaks
\onehalfspacing
\title{Genuinely Robust Inference for Clustered Data\thanks{First arXiv date: August 23, 2023. This version: September 22, 2025. {\color{black}We owe special thanks to A. Colin Cameron, Tetsuya Kaji, and Kevin Song for their invaluable comments that significantly improved our methodology. We also benefited from discussions with Ivan Canay, Bruce Hansen, Michael Jansson, Seojeong (Jay) Lee, Michael Leung, James MacKinnon, Francesca Molinari, Ulrich M\text{\"u}ller, Jack Porter, Frank Schorfheide, Xiaoxia Shi, Max Tabord-Meehan, Tim Vogelsang, and seminar participants at numerous institutions and conferences.} We  thank Hong Xu and Siyuan Xu for their excellent research assistance. All remaining errors are ours.  
}}
\author{Harold D. Chiang\thanks{Assistant Professor of Economics, University of Wisconsin-Madison. \texttt{hdchiang@wisc.edu}}
\and
Yuya Sasaki\thanks{
Brian and Charlotte Grove Chair and Professor of Economics, Vanderbilt University. Email: \texttt{yuya.sasaki@vanderbilt.edu}} 
\and 
Yulong Wang\thanks{Associate Professor of Economics, Syracuse University. Email: \texttt{ywang402@syr.edu}}}
\date{}
\maketitle

\begin{abstract}
\setlength{\baselineskip}{6.67mm}
Conventional cluster-robust inference can be invalid  when data contain clusters of unignorably large size. We formalize this issue by deriving a necessary and sufficient condition for its validity, and show that this condition is frequently violated in practice: specifications from 77\% of empirical research articles in \textit{American Economic Review} and \textit{Econometrica} during 2020–2021 appear not to meet it. To address this limitation, we propose a genuinely robust inference procedure based on a new cluster score bootstrap. We establish its validity and size control across broad classes of data-generating processes where conventional methods break down. Simulation studies corroborate our theoretical findings, and empirical applications illustrate that employing the proposed method can substantially alter conventional statistical conclusions.

{\small { \ \ \newline
}
}
\end{abstract}

\setlength{\baselineskip}{7.45mm}
\newpage
\sloppy

\section{Introduction}\label{sec:introduction}

Cluster-robust (CR) standard errors are designed to account for within-cluster correlations. 
Such correlations often arise by construction, for example, within an industry \citep{He1998} or within a state \citep{BeDuMu2004}. 
Today, even when a model does not inherently induce cluster dependence, the application of CR methods using observable group identifiers has become a common practice.

The foundational theory \citep*[][]{Wh84,LiZe86,Ar87} for CR inference methods assumes small cluster sizes $N_g$ (uniformly bounded above by $\overline{N} < \infty$) with a large number of clusters, $G \rightarrow \infty$, where $N_g$ denotes the number of entities in the $g$-th cluster for $g \in \{1,2,\dots,G\}$. 
Procedures based on this theory are implemented through the `\texttt{cluster()}' and `\texttt{vce(cluster)}' options in Stata, and they are utilized in nearly all, if not all, empirical studies that report CR standard errors.

It has been recognized that large cluster sizes $N_g$ can result in inflated CR standard errors \citep*[e.g.,][p. 324]{CaMi15}. Recent theoretical advancements \citep[][]{carter2017asymptotic,DjMaNi19,HaLe19,Ha22,BuCaShTa2022} accommodate larger cluster sizes $N_g$, eliminating the requirement that $N_g \leq \overline{N}$ and thereby broadening the applicability of the `\texttt{cluster()}' and `\texttt{vce(cluster)}' options, among others.
With this said, they still impose the restriction $\max_g N^2_g / N \rightarrow 0$ of vanishing maximum cluster size relative to the square root of the whole sample size $N = \sum_{g=1}^G N_g$ as $G \rightarrow \infty$. 

A natural question is whether the relaxed condition $\max_g N_g^2 / N \rightarrow 0$ accommodates a wide range of data sets. To answer this, we analyze empirical papers published in top journals.\footnote{We studied all articles published in \textit{American Economic Review} and \textit{Econometrica} between 2020 and 2021. Among them, we extracted a list of papers reporting estimation and inference results based on regressions, IV regressions, and their variants. Furthermore, we focus on articles using publicly available data sets for replication. See Section \ref{sec:fragility} for further details of this study.} All of these articles employ the aforementioned Stata options for CR standard errors, thereby implicitly assuming $\max_g N_g^2 / N \rightarrow 0$. Table \ref{tab:bins} summarizes the number of articles with $\max_g N_g^2 / N$ falling into each bin on a logarithmic scale. Notably, 55 percent (respectively, 39, 29, and 16 percent) of the articles use data sets where $\max_g N_g^2 / N \geq 1$ (respectively, $\geq 10$, $\geq 100$, and $\geq 1000$). In other words, the condition $\max_g N_g^2 / N \rightarrow 0$, required for the validity of conventional CR inference, may not hold for a nontrivial portion of these published articles.

\begin{table}[t]
\centering
\begin{tabular}{lcccccc}
\multicolumn{7}{c}{The Empirical Distribution of $\max_g N^2_g / N$ in Empirical Economic Research: 2020--2021}\\
\hline\hline
\multirow{2}{*}{$\max_g N^2_g / N$}&&0.1--&1--&10--&100--&$\geq$1000 \\ 
&$<$0.1&1&10&100&1000&\\
\hline
\textit{American Economic Review}& 4 & 8 & 4 & 1 & 3 & 1\\
\textit{Econometrica}            & 2 & 0 & 1 & 2 & 1 & 4\\
\multicolumn{1}{r}{Total}        & 6 & 8 & 5 & 3 & 4 & 5\\
& \ (19\%) \ & \ (26\%) \ & \ (15\%) \ & \ (10\%) \ & \ (13\%) \ & \ (16\%) \ \\
\hline\hline
\end{tabular}
\caption{Distribution of articles by the value of $\max_g N^2_g / N$, classified into the bins $[0,0.1)$, $[0.1,1)$, $[1,10)$, $[10,100)$, $[100,1000)$, and $[1000,\infty)$ on a logarithmic scale. The sample consists of papers published in the \textit{American Economic Review} and \textit{Econometrica} during 2020--2021 that report cluster-robust standard errors for regression or IV regression using publicly available replication data sets. For papers reporting multiple regressions, we record the largest value of $\max_g N^2_g / N$ across specifications.}
\label{tab:bins}
\end{table}

The condition $\max_g N_g^2 / N \to 0$ is sufficient but not necessary for asymptotic normality, implying that the adequacy of normality-based confidence intervals and tests cannot be evaluated solely by assessing the plausibility of this condition. To address this, we establish a necessary and sufficient condition for the validity of conventional cluster-robust (CR) inference -- see Theorem \ref{thm:iff}. Specifically, the limiting distribution is normal if and only if the score of the largest cluster is ignorable. When clusters are \textit{unignorably large}, regression estimates exhibit non-Gaussian limiting distributions, as illustrated in Figure \ref{figure:limit_distribution}.\footnote{Details on these non-Gaussian distributions are provided in Section \ref{sec:fragility}.} Using this characterization, formal statistical tests based on \citet{sasaki2023diagnostic} reject the null hypothesis of normality in 24 of the 31 papers (77 percent) reported in Table \ref{tab:bins}--see Table \ref{tab:test_alpha}.

\begin{figure}
\centering
\includegraphics[width=0.32\textwidth]{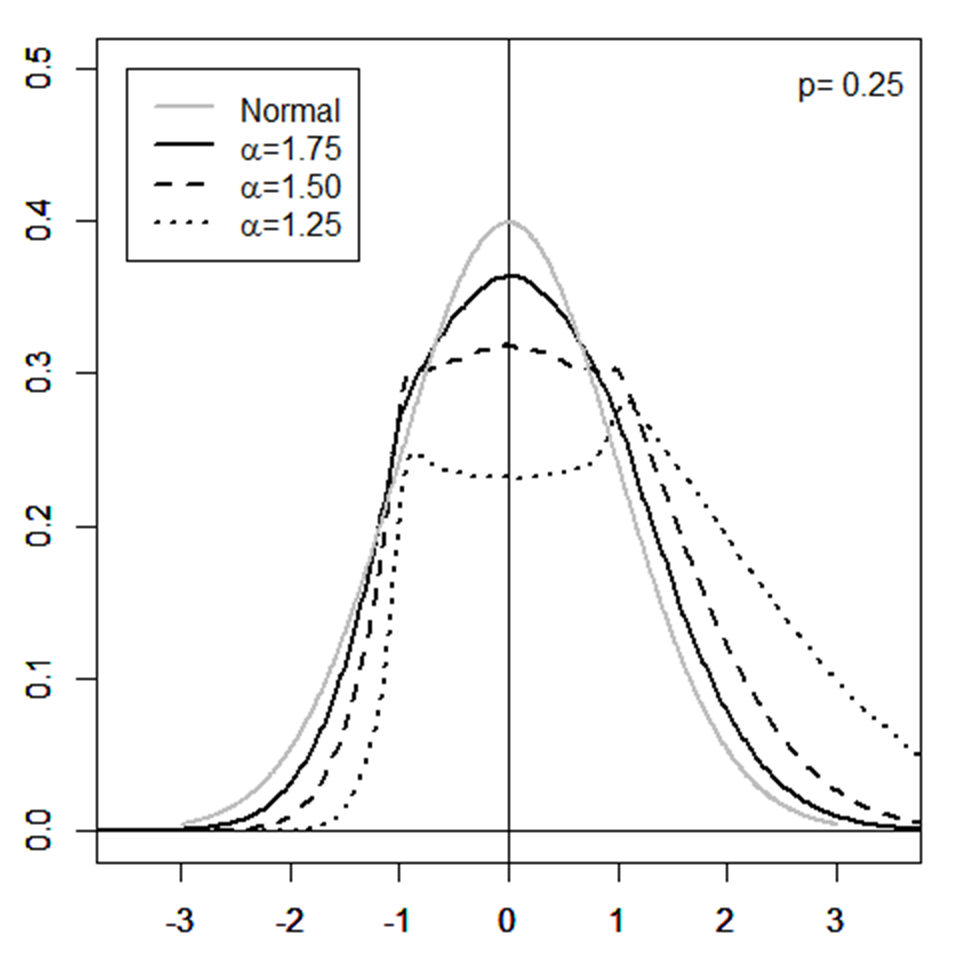}
\includegraphics[width=0.32\textwidth]{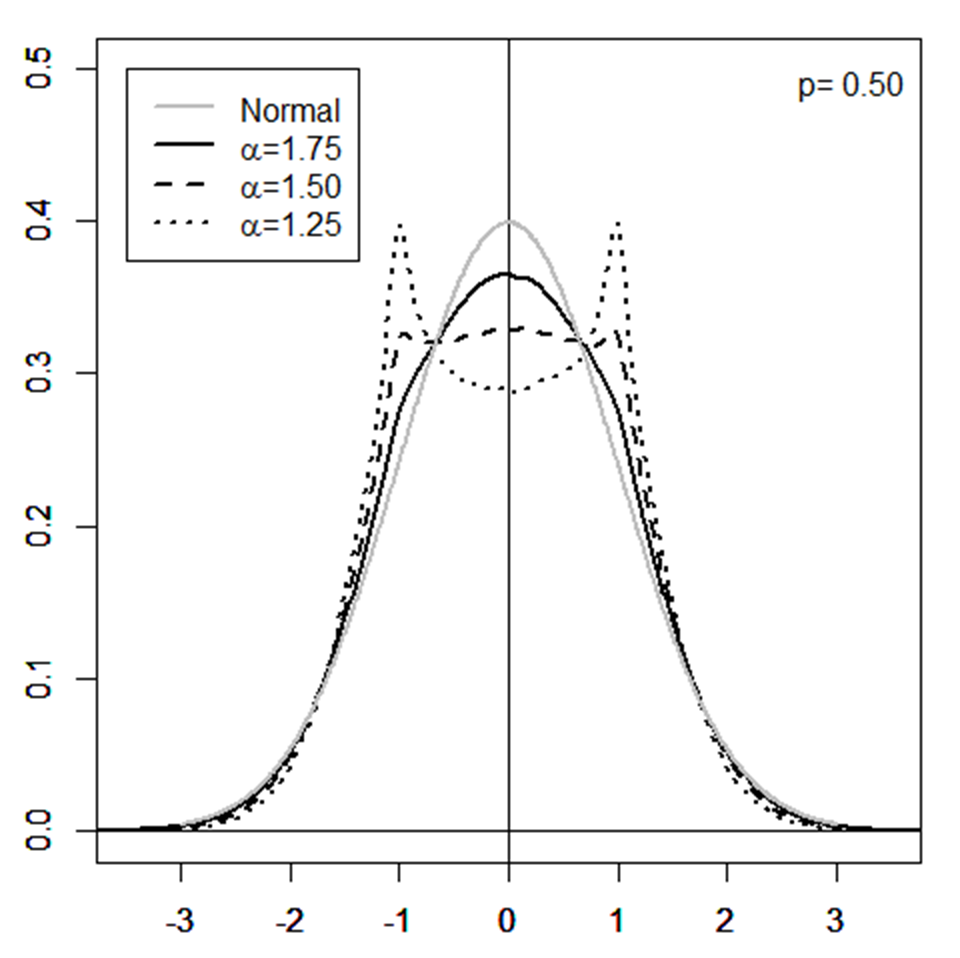}
\includegraphics[width=0.32\textwidth]{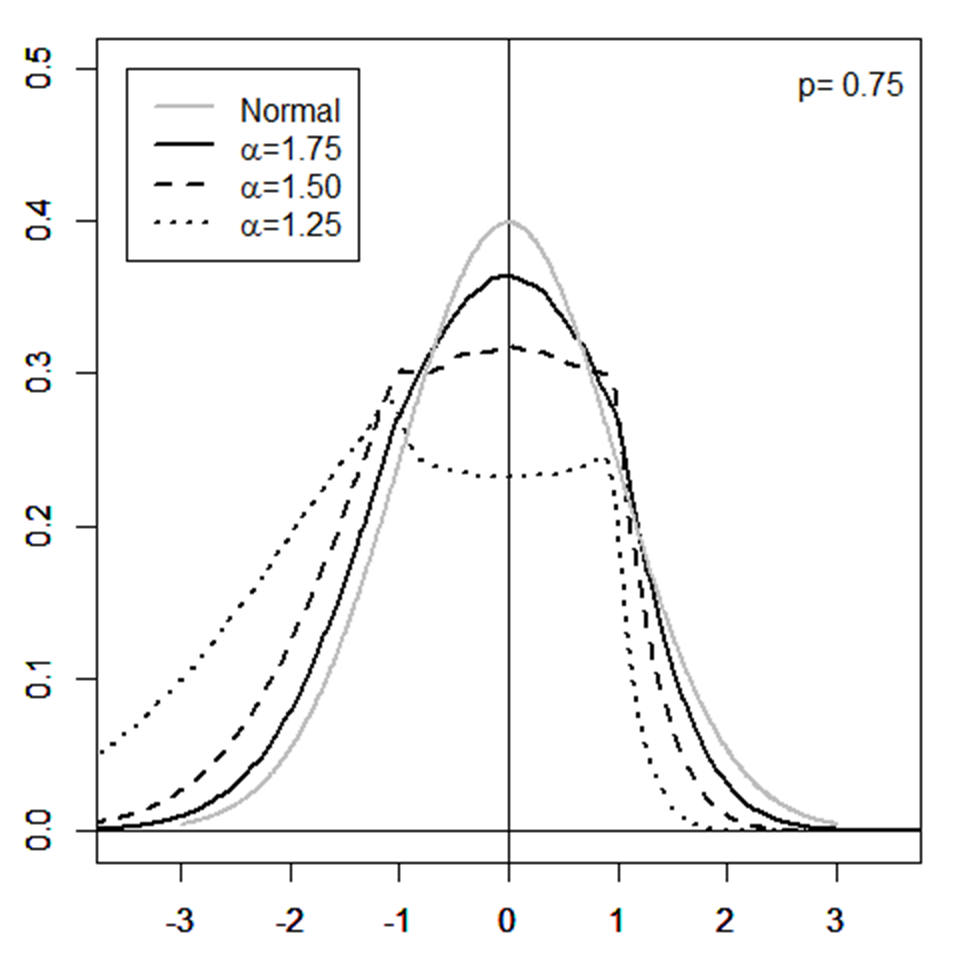}
\caption{Illustration of non-Gaussian limiting distributions arising from unignorably large clusters. The different distributional shapes, indexed by $\alpha$ and $p$, are described in Section~\ref{sec:fragility}.}
\label{figure:limit_distribution}
\end{figure}

Non-Gaussian limiting distributions invalidate conventional critical values, such as ``1.96,'' as well as bootstrap critical values. For example, using 1.96 results in sizes of 0.053, 0.087, and 0.250 (instead of the desired 0.050) when the nuisance parameter $\alpha$ equals 1.75, 1.50, and 1.25, respectively, as shown in Figure \ref{figure:limit_distribution}. The empirical bootstrap fails in these cases of infinite variance, and the widely used wild cluster bootstrap and pairs cluster bootstrap are also inconsistent.
Later, we formally establish these negative results as Proposition \ref{prop:boot}.

To address this issue, we introduce the cluster score (CS) bootstrap, a novel inferential method for clustered data that extends the $m$-out-of-$n$ and score bootstraps.
This method provides valid critical values adaptively across all limiting distributions depicted in Figure \ref{figure:limit_distribution}. 
In addition, we provide a data-driven choice of the tuning parameter and justify its theoretical validity.

\medskip

\noindent{\bf Relation to the Literature:}
The literature on cluster-robust inference has a long history, dating back to \citet{Wh84}, \citet{LiZe86}, and \citet{Ar87}. 
For a comprehensive review, we refer readers to \citet{CaMi15,cameronreview} and \citet{mackinnon2023cluster}.
More recently, sampling frameworks in which cluster sizes are treated as random variables have been investigated by \citet{bai2022inference}, \citet{BuCaShTa2022}, and \citet{cavaliere2022econometrics}.  
We adopt a model-based perspective with an increasing number of clusters and unrestricted intra-cluster dependence, or, asymptotically equivalently, a sampling-based perspective where the growing number of sampled clusters represents a negligible fraction of the superpopulation. This framework is well suited to the empirical contexts encountered in most applications.\footnote{
An alternative framework assumes a fixed number of clusters with growing cluster sizes, where asymptotic normality can be derived under additional assumptions of weak intra-cluster dependence, as in \citet{Ha07}, \citet{IbMu10}, \citet{canay2021wild}. \citet{ibragimov2016inference} and \citet{hansen2022jackknife} consider inference under gaussian assumptions. 
Another line of research advances the integration of design-based and sampling-based asymptotics, particularly under explicit treatment assignment schemes such as randomized experiments, as considered by \citet{AbAtImWo23}. Extending our method to this design-based framework is an avenue for future research.
}

In an insightful recent work, \citet{kojevnikov2021some} establish an impossibility result for consistent estimation of the asymptotic variance when the sample contains a single large cluster under a triangular array setup. 
They further provide a necessary and sufficient condition on the cluster structure for the asymptotic variance to be consistently estimable. 
Our findings complement their result by showing that normal approximation for $t$-statistics fails in the presence of unignorably large clusters. 
Furthermore, our proposed procedure overcomes this limitation, as it does not rely on consistent variance estimation. 
We demonstrate that the self-normalized statistic converges in distribution and formally derive its limiting stable distribution in such settings. 
Importantly, the implementation of the bootstrap inference procedure does not require knowledge of the unknown rate or consistent variance estimation, owing to the self-normalizing nature of the test statistics.

The aspect of our paper that establishes non-Gaussianity under certain conditions connects to a branch of the econometrics literature exploring the possibility of non-Gaussian limiting distributions and, in some cases, impossibility results. 
For instance, \citet{hirano2012impossibility} show that regular estimation, and hence Gaussian asymptotics, are impossible for a class of estimators characterized by the maximum. 
More directly related, \citet{menzel2021bootstrap} highlight the potential non-Gaussianity of estimators under two-way clustering and establish an impossibility result on the uniform consistency of tests. 
In contrast, we demonstrate that non-Gaussianity can arise even under one-way clustering. 
Moreover, such negative results may be even more pervasive in this widely used empirical setting.

Our key distributional approximation results build on \citet{logan1973limit}, \citet{lepage1981convergence}, and \citet{gine1997student}. 
For theoretical foundations of probability and statistics with heavy-tailed distributions, we refer readers to \citet{resnick1987extreme,resnick2008extreme} and \cite{samorodnitsky1994stable}. Also see \cite{ibragimov2015heavy,chernozhukov2016extremal,chernozhukov2017extremal} for applications in economics and finance.
Our inference procedure builds on the resampling theory developed in \citet{arcones1989bootstrap,arcones1991additions} and \cite{bickel2008choice}. For  the related discussions on the inconsistency of empirical bootstrap for means of random variables with infinite variance, see, e.g., \citet{athreya1987bootstrap} and \citet{knight1989bootstrap}.


\section{The Model}\label{sec:model}

While the idea extends to a general class of econometric models, 
we consider the linear model\footnote{The assumption $\E[U_{gi}|X_g]=0$, while standard in the literature, is stronger than required for our asymptotic results. It can be relaxed to $\E[\sum_{i=1}^{N_g}X_{gi}U_{gi}]=0$.}
\begin{align*}
Y_{gi} = X_{gi}'\theta +U_{gi}, \qquad \E[U_{gi}|X_g]=0 \quad i=1,...,N_g,
\end{align*}
for ease of exposition as well as its popular use in practice,
where $X_g = (X_{g1},\ldots,X_{gN_g})'$, $U_g = (U_{g1},\ldots,U_{gN_g})'$, $g \in \{1,\ldots,G\}$ indexes clusters, and $N_g$ denotes the size of the $g$-th cluster. 
Define the OLS estimator and its cluster-robust (CR) variance estimator by
\begin{align}
&\hat \theta = \left(\sum_{g=1}^G \sum_{i=1}^{N_g}X_{gi}X_{gi}'\right)^{-1}\sum_{g=1}^G \sum_{i=1}^{N_g}X_{gi}Y_{gi}=\left(\sum_{g=1}^G X_g'X_g \right)^{-1} \sum_{g=1}^G  (X_g'X_g\theta + S_g)\quad\text{ and } \label{eq:ols} \\
&\hat V^{\text{CR}}=a_G \left(\sumg X_g'X_g\right)^{-1}\left(\sumg \hat S_g\hat S_g'\right)\left(\sumg X_g'X_g\right)^{-1}, \label{eq:cr}
\end{align}
respectively, for some finite sample adjustment factor $a_G$ such that $a_G \to 1$ as $G \to \infty$,
where $S_g = \sum_{i=1}^{N_g} X_{gi} U_{gi}$, $\hat S_g = \sum_{i=1}^{N_g} X_{gi} \hat U_{gi}$, and $\hat U_{gi}=Y_{gi}-X_{gi}'\hat\theta$. 
For simplicity of writing, we set $a_G=1$ throughout as it does not affect our asymptotic arguments.

Consider a linear transformation $\delta=r'\theta$ of the regression coefficient vector $\theta$, such that $r\in \Real^{\dim(\theta)}$ and $\|r\|=1$, as the parameter of interest. 
Let the corresponding estimator and its CR standard error be denoted by
\begin{align*}
\hat \delta =& \, r'\hat \theta 
\quad\text{ and }\\
\hat \sigma^2 =& \, r' \left(\sumg X_g'X_g\right)^{-1}\left(\sumg \hat S_g\hat S_g'\right)\left(\sumg X_g'X_g\right)^{-1} r,
\end{align*}
respectively.
We are interested in conducting inference for $\delta$ using the t-statistic
\begin{align}\label{eq:t_statistic}
\frac{(\hat \delta- \delta)}{\hat\sigma}=\frac{r'(\hat \theta-\theta)}{ \sqrt{r' \left(\sumg X_g'X_g\right)^{-1}\left(\sumg \hat S_g\hat S_g'\right)\left(\sumg X_g'X_g\right)^{-1} r}}
\end{align}
based on the CR standard error.

To state our assumption, we introduce a few definitions.
A random variable $\eta$ is said to be \textit{stable} if it has a domain of attraction in that there exists a sequence of i.i.d. random variables $\xi_1,\xi_2,\ldots$ and sequences of positive numbers $A_G$ and real numbers $D_G$ such that
\begin{align*}
	\frac{\sumg \xi_g - D_G}{A_G}\stackrel{d}{\to}\eta \qquad\text{as } G\to \infty.
\end{align*} 
A function $L(\cdot)$ is said to be \textit{slowly varying} at $\infty$ if $\lim_{t\to \infty}L(yt)/L(t)=1$ for all $y>0$.
If $\eta$ is stable, then $A_G$ takes the form of $G^{1/\alpha}L(G)$ for some $\alpha\in (0,2]$ and some slowly varying function $L(\cdot)$ at $\infty$ (cf. Proposition 2.2.13 in \citealt{embrechts1997modelling}).
If $\alpha \in (1,2]$, then $D_G$ can be chosen to be $G\cdot \E[\xi_g]$.
The number $\alpha$ is called the \textit{index of stability}, and $\eta$ is said to be \textit{$\alpha$-stable}.
In such a case, $\xi_g$ is said to belong to the \textit{domain of attraction} of an \textit{$\alpha$-stable distribution}.
Although this concept may look esoteric to some readers, it essentially states that a sum of i.i.d. random variables, after being suitably centered and normalized, converges in distribution to a limiting random variable, and it, in particular, encompasses the standard cases where central limit theorems (CLTs) hold.
In other words, econometricians and economists adopting the standard inference (e.g., the conventional critical value of 1.96) implicitly make this (and even stronger) assumption.

\begin{assumption}\label{assm}
$(X_g'X_g,S_g)_{g=1}^G$ are i.i.d., follow a non-degenerate distribution, $\E[N_g]=c\in (0,\infty)$, and the design matrix satisfies
$G^{-1}\sumg X_g'X_g=Q + o_p(1)$
for a finite positive definite matrix $Q$.
For $v=r'Q^{-1}$ and for all $u_1,u_2\in \Real^{\dim(\theta)}$ with unit length, $v'S_g$ and $u_1'X_g'X_gu_2$ belong to the domain of attraction of stable laws with an index of stability $\alpha \in (1,2]$.
\end{assumption}

This assumption is arguably general. 
It is significantly weaker than requiring a central limit theorem to hold and even covers scenarios where the asymptotic normality fails. 
It also encompasses two notable cases frequently considered in economics and econometrics.

First, the case of $\alpha = 2$ encompasses the conventional setting in which $r'(\hat\theta - \theta)$ attains the standard convergence rate of $\sqrt{G}$ by the central limit theorem. 
In this case, the limiting $\alpha$-stable distribution is necessarily Gaussian \citep[cf.][Theorem~2]{Geluk2000}. 
This scenario also includes certain non-standard cases with a Gaussian limiting distribution but without finite variance, such as a Pareto random variable with Pareto exponent equal to $2$. 
The vast majority of econometric papers deriving asymptotic normality implicitly rely on this high-level condition in Assumption~\ref{assm}, or on even stronger ones.

Second, the case of $\alpha < 2$ entails a power law \citep[][Theorem~2.24]{pena2009self}, i.e.,
\begin{align}\label{eq:power}
P(|v'S_g| > t) &= t^{-\alpha} L_1(t), 
\qquad
P(|u_1 X'_g X_g u_2| > t) = t^{-\alpha} L_{2}(t),
\end{align}
for some slowly varying functions $L_1(\cdot)$ and $L_2(\cdot)$, where $L_2(\cdot)$ may depend on $u_1$ and $u_2$. 
In this case, the index $\alpha$ of stability coincides with the Pareto exponent\footnote{Specifically, the Pareto distribution has CDF $F(t) = 1 - t^{-\beta}$ for $t \geq 1$.} $\beta$, in the sense that $\alpha = \min\{\beta,2\}$. 
Thus, when $\alpha < 2$, the score has infinite variance. 
For more precise details, see Theorem~\ref{theorem:pena} in Appendix~\ref{sec:alternative}.  
In this case, \emph{unignorably large} clusters are indeed unignorable, since the sample sum of the (scaled) scores becomes asymptotically proportional to the (scaled) score of the largest cluster; see Remark~\ref{remark:impossibility} in Appendix~\ref{sec:theory_alpha_less_than_two} for further discussion. 
Hence, the asymptotic distribution cannot be Gaussian in this case.


The literature in urban economics and economic geography establishes that (truncated) city size distributions frequently exhibit a power-law behavior in the upper tail, with estimated exponents typically in the range of $1$ to $1.5$; see \citet{eeckhout2004gibrat,ioannides2013us} and references therein. Consequently, when observations are strongly correlated at the city level, this implies $\alpha < 2$, yielding a non-Gaussian limiting distribution.

The i.i.d. requirement across clusters in Assumption~\ref{assm} is standard in this literature (e.g., \citealt{BuCaShTa2022,cavaliere2022econometrics,bai2022inference}). 
This assumption is mild because:  
(1) the conditional distributions of $S_g$ and $X'_g X_g$ given $N_g = n_g$ may vary across $n_g$; and  
(2) the distributions of individuals within each cluster need not be identical.  
Moreover, $S_g$ and $X_g$ may be arbitrarily correlated with the cluster size $N_g$, provided that the regression exogeneity condition is satisfied.

To simplify the exposition, we focus on the case where $v'S_g$ and $u_1' X_g' X_g u_2$ share a common stability index $\alpha$. 
This simplification is rationalized if the tail behavior of their distributions is driven by the tail behavior of the cluster-size distribution $N_g$; see Section~\ref{sec:fragility:heuristic} for an illustrative example. 
However, this simplification is adopted only for notational simplicity and can be relaxed at the cost of substantially more cumbersome exposition.

\section{Fragility of the Conventional CR Methods}\label{sec:fragility}
This section shows that conventional methods of cluster-robust (CR) inference are valid if and only if $\alpha = 2$. 
In other words, they necessarily fail when $\alpha < 2$.
We begin with heuristic discussions in Section~\ref{sec:fragility:heuristic} and then develop formal results in Section~\ref{sec:fragility:theory}. 
We further report how frequently cases with $\alpha < 2$ arise in empirical economic applications.

\subsection{Heuristic Discussions}\label{sec:fragility:heuristic}
The intuition behind the fragility of conventional CR methods is as follows. 
When $\alpha < 2$, the cluster size $N_g$ does not have a finite variance. 
If intra-cluster dependence is non-trivial, this infinite variance of $N_g$ is inherited by the score $S_g$, causing the CLT for OLS (and also other estimators) to fail. 
Cases with $\alpha < 2$ are quite plausible in empirical data, and we show that this is indeed the case for the majority of recent empirical papers published in \textit{Econometrica} and the \textit{American Economic Review}, as discussed in Section~\ref{sec:fragility:theory} in detail.

To provide a simple and transparent illustration, consider the sample average
\begin{equation*}
\hat\theta = \frac{1}{N} \sum_{g=1}^G \sum_{i=1}^{N_g} Y_{gi},
\end{equation*}
which is a special case of the OLS estimator \eqref{eq:ols} with $X_{gi} = 1$. 
The true parameter is the mean $\theta = \mathbb{E}[Y_{gi}]$, which is normalized to $\theta = 0$ without loss of generality. 
For clarity, suppose the extreme case of perfect intra-cluster dependence, i.e., $Y_{gi} \equiv Y_{g}$ for all $i \in \{1, \ldots, N_g\}$ within each cluster $g$. 
For simplicity, also assume that $N_g$ is independent of $Y_g$. 
These assumptions are made purely for expositional clarity and are not essential for our results.

In this case, we obtain
\begin{equation*}
\sqrt{N}\hat{\theta} 
= \frac{G^{-1/2}\sum_{g=1}^{G} N_{g} Y_{g}}
       {\sqrt{\tfrac{1}{G}\sum_{g=1}^{G} N_{g}}}.
\end{equation*}%
The denominator converges to $\sqrt{\mathbb{E}[N_{g}]}$ as long as $\alpha > 1$. 
For the numerator, note that 
$
\Var\!\left(G^{-1/2}\sum_{g=1}^{G} N_{g} Y_{g}\right) 
= \Var[N_g] \cdot \Var[Y_g],
$
which is infinite when $\alpha < 2$. 
Indeed, Theorem~1 of \citet{Geluk2000} implies that, if the distribution of $N_g Y_g$ is $\alpha$-stable, then the limiting distribution
\begin{equation*}
x \;\mapsto\; 
\lim_{G \to \infty} 
\Pr\!\left( \frac{1}{a_{G}} \sum_{g=1}^G N_g Y_g - b_G > x \right),
\end{equation*}
for some sequences $a_G \simeq G^{1/\alpha} \to \infty$ and $b_G \in \mathbb{R}$, 
has the characteristic function
\begin{equation*}
\psi_{\alpha}(s) 
= \exp \left\{ 
   -\Big( |s|^{\alpha} 
   + i s (1-\alpha) \tan(\alpha\pi/2) \tfrac{|s|^{\alpha - 1}-1}{\alpha - 1} 
   \Big) 
 \right\},
\end{equation*}
which differs from the Gaussian characteristic function. 
Thus, the CLT for $\hat{\theta}$ fails. 
Further discussion of this example is provided in Appendix~\ref{sec:appendix:self_normalized}.

In summary, the stability index $\alpha$ determines both the convergence rate and the limiting distribution. 
When intra-cluster correlation is non-trivial, the tail heaviness of $N_g$ carries over to that of $S_g$. 
Consequently, the $t$-ratio of the conventional CR method is asymptotically normal \textit{if and only if} $\alpha = 2$. While we currently present heuristic arguments in a simplified setting, we formalize and generalize this claim in Theorem~\ref{thm:iff} and Proposition~\ref{prop:boot} below.




\begin{remark}[Bias from Trimming Large Clusters]
In practice, researchers may trim large clusters with $N_g > k$ for some threshold $k$. 
While such trimming may appear to mitigate problems arising from non-Gaussian limiting distributions induced by unignorably large clusters, it introduces bias and thereby undermines the validity of inference. 
Consider
\[
0 = \mathbb{E}\!\left[ \sum_{i=1}^{N_g} Y_{gi} \right] 
  = \mathbb{E}\!\left[ k Y_{g} \,\mathds{1}\{N_g \leq k\} \right] 
  + \mathbb{E}\!\left[ (N_g - k) Y_{g} \,\mathds{1}\{N_g > k\} \right] 
  =: \theta(k) + \lambda(k).
\]

Here, $\theta(k)$ represents the estimand of the trimmed procedure, while $\lambda(k)$ denotes the associated bias term. 
This bias $\lambda(k)$ is generally nonzero whenever the distribution of $Y_{g}$ depends on $N_g$. 
Hence, naively trimming large clusters can result in invalid inference.
$\blacktriangle$
\end{remark}

\subsection{Formal Theory}\label{sec:fragility:theory}

The following theorem formalizes and generalizes the discussion from the previous subsection.
\begin{theorem}[Necessary and sufficient condition]\label{thm:iff}
If Assumption \ref{assm} is satisfied for an $\alpha \in (1,2]$, then the t-statistic \eqref{eq:t_statistic} is asymptotically normal if and only if $\alpha=2$.
\end{theorem}

\noindent
A proof is provided in Appendix~\ref{sec:proof:thm:iff}. 

The theorem implies that conventional CR inference based on common variance estimators, such as CR1, CR2, CR3, and the jackknife, together with normal critical values (e.g., $\approx 1.96$ for the 97.5th percentile) fails whenever $\alpha < 2$.

With Theorem~\ref{thm:iff}, we now characterize the curves shown in Figure~\ref{figure:limit_distribution} from Section \ref{sec:introduction}. 
The left, middle, and right panels of Figure~\ref{figure:limit_distribution} display the limiting distributions of the $t$-statistic under $p = 0.25$, $0.50$, and $0.75$, respectively, where 
\begin{align}
p = \lim_{t \to \infty} \frac{P\!\left(v' S_g > t\right)}{P\!\left(\lvert v' S_g \rvert > t\right)}, \label{eq:tail_balancing}
\end{align}
for $v$ as given in Assumption~\ref{assm}, measures the limiting asymmetry of tail probabilities. 
Each panel in Figure~\ref{figure:limit_distribution} depicts three non-Gaussian limiting distributions corresponding to $\alpha = 1.25$, $1.50$, and $1.75$ with distinct line styles, together with the normal reference case ($\alpha = 2.00$). 
The key takeaway is that conventional methods of CR inference, which rely on normal approximation, become increasingly size-distorted as $\alpha$ decreases and as $p$ deviates from $0.5$.

Another class of conventional approaches consists of cluster bootstraps. 
Two main bootstrap-based CR inference methods are commonly used in the literature: the pairs cluster bootstrap and the wild cluster bootstrap \citep{cameron2008bootstrap}. 
It is well established that the empirical bootstrap is inconsistent when the variance of the score is infinite \citep[cf.][]{athreya1987bootstrap,knight1989bootstrap}. 
In light of the power-law characterization \eqref{eq:power}, the pairs cluster bootstrap, essentially the empirical bootstrap applied to cluster-wise sums treated as independent units, is inconsistent under Assumption~\ref{assm} with $\alpha < 2$. 
Moreover, Theorem~\ref{theorem:bootstrap} in Appendix~\ref{sec:inconsistent_wcb} shows that the wild cluster bootstrap is likewise inconsistent under Assumption~\ref{assm} with $\alpha < 2$. 
The following proposition summarizes these results.

\begin{proposition}[Failure of the Conventional Cluster Bootstraps]\label{prop:boot}
If Assumption \ref{assm} is satisfied for an $\alpha<2$, then the pairs cluster bootstrap and the wild bootstrap methods are both inconsistent. 
\end{proposition}

Given that the case of $\alpha < 2$ invalidates all conventional methods of CR inference, a natural question is how frequently such cases arise in empirical economics. 
To address this, we examined all articles published in two leading journals, the \textit{American Economic Review} and \textit{Econometrica}, during 2020–2021. 
From these, we extracted the subset of papers reporting estimation and inference results based on regressions, IV regressions, and related variants. 
We further restricted attention to articles that employ publicly available datasets due to replicability.

For these articles, we test the null hypothesis $H_0: \alpha = 2$ against the alternative $H_1: \alpha < 2$ for the score. 
Such a test can be implemented via the likelihood ratio test of \citet{sasaki2023diagnostic}, which considers the surrogate null hypothesis $H_0: \beta \ge 2$ against the alternative $H_1: \beta < 2$ in light of \eqref{eq:power}, where $\beta$ denotes the tail exponent of the score.\footnote{The test of the null hypothesis $H_0: \beta \ge 2$ against the alternative $H_1: \beta < 2$ is implemented using the Stata command 
\texttt{testout y x1 x2 ..., cluster(cid)} for least-squares estimation and 
\texttt{testout y x1 x2 ..., iv(z) cluster(cid)} for instrumental variables estimation, 
both following \citet{sasaki2023diagnostic}.}

Table~\ref{tab:test_alpha} summarizes the set of papers included in our study. 
The first two columns report the journals and years of publication. 
The next column, ``All \#,'' indicates the total number of eligible articles according to the selection criteria described above. 
The column group labeled ``Cluster'' contains articles in which CR inference is applied to at least one regression result. 
Within this group, the column ``\#'' reports the number of such articles, while the column ``Test $\alpha < 2$'' reports the fraction of these articles for which the test rejects the null hypothesis in one or more regression specifications. 
The final row presents the column totals.

\begin{table}[t]
\begin{center}
\begin{tabular}{lcccrrcccrr}
\hline\hline
&Year of& All && \multicolumn{3}{c}{Cluster}\\
\cline{3-3}\cline{5-7}
Journal & Publication & \# && \# & \multicolumn{2}{c}{Test $\alpha<2$}\\
\hline
\textit{American Economic Review} & 2020 & 15 && 10 & 7/10 & (70\%) \\
\textit{American Economic Review} & 2021 & 15 && 11 & 9/11 & (82\%) \\
\multicolumn{2}{r}{Subtotal} & 30 && 21 & 16/21 & (76\%)\\
\\
\textit{Econometrica} & 2020 & 12 && 7 & 7/8 & (88\%)\\
\textit{Econometrica} & 2021 & 3 && 2 & 1/2 & (50\%)\\
\multicolumn{2}{r}{Subtotal} & 15 && 10 & 8/10 & (80\%)\\
\\
\multicolumn{2}{r}{Total} & 45 && 31 & 24/31 & (77\%)\\
\hline\hline
\end{tabular}
\caption{The column ``All -- \#'' reports the total number of eligible articles that use regressions or IV regressions with publicly available replication data. The column ``Cluster -- \#'' reports the number of these articles that employ cluster-robust (CR) inference. The column ``Cluster -- Test $\alpha<2$'' reports the rejection rate of the null hypothesis $\alpha=2$ among articles using CR inference. The tests of $\alpha=2$ against the alternative $\alpha<2$ are implemented with the Stata commands ``\texttt{testout y x1 x2 ..., cluster(cid)}'' for regressions and ``\texttt{testout y x1 x2 ..., iv(z) cluster(cid)}'' for IV regressions, following \citet{sasaki2023diagnostic}.}${}$
\label{tab:test_alpha}
\end{center}
\end{table}

During 2020--2021, the \textit{American Economic Review} published 30 articles that met our selection criteria. 
Of these, 21 reported CR standard errors. 
The null hypothesis is rejected in 16 of these 21 papers. 
In other words, inference based on the conventional CR method may be misleading in approximately $76\%$ of the articles employing it.

During 2020--2021, \textit{Econometrica} published 14 articles that met our selection criteria. 
Of these, 10 reported CR standard errors. 
The null hypothesis is rejected in 8 of these 10 papers. 
In other words, inference based on the conventional CR method may be misleading in approximately $80\%$ of the articles employing it.

Combining the two journals, we find that inference may be misleading in as many as $77\%$ of the 31 articles that employ the conventional CR method. 
Thus, problematic practice appears to be prevalent even in these highly influential outlets.\footnote{Spreadsheets of all the test results with specific papers and specific equations are available upon request.} 
All of the above issues with conventional CR methods motivate our proposed approach: the cluster score (CS) bootstrap, which accommodates non-Gaussian limiting distributions, to be presented in Section~\ref{sec:bootstrap}.

\section{The Cluster Score Bootstrap}\label{sec:bootstrap}

In light of the limitations of conventional CR inference methods discussed in the previous section, 
we introduce a novel cluster score (CS) bootstrap procedure to approximate the limiting distribution of 
$(\widehat{\delta} - \delta)/\widehat{\sigma}$. 
This procedure remains valid whether the limiting distribution is Gaussian or non-Gaussian. 
Section~\ref{sec:method} describes the proposed method, and Section~\ref{sec:theory} provides its theoretical justification.

\subsection{The Method}\label{sec:method}

Our objective is to conduct statistical inference for $\delta$ using the $t$-statistic defined in 
\eqref{eq:t_statistic}. 
Let the CDF $J_G^\ast$ of the $t$-statistic be
\begin{align*}
J_G^*(t) = P\!\left( \frac{\widehat{\delta} - \delta}{\widehat{\sigma}} \leq t \right).
\end{align*}
We will show that, under suitable regularity conditions, $J_G^*$ converges to the CDF $J^*$ of the corresponding limiting distribution.

Let $b$ denote the number of resampled clusters, chosen according to 
Algorithm~\ref{alg:bickel_sakov} (to be presented in Section~\ref{sec:theory}). 
For a large positive integer \(M\), draw \(M\) i.i.d. multinomial random vectors \((w_1^j,\ldots,w_G^j)_{j=1}^M\), each with \(b\) trials and uniform cell probability \(1/G\), independently of the data. This is equivalent to sampling \(b\) clusters  with replacement uniformly from the \(G\) clusters in the data, 
and \(w_g^j\) records the counts of how many times the $g$-th cluster is selected in the $j$-th bootstrap sample.
Define the CS bootstrap estimator and its associated variance estimator by
\begin{align*}
\widehat{\delta}_{b,j} 
&= r'\widehat{\theta}_{b,j} 
= \left(\frac{G}{b}\right) 
   r' \left( \sum_{g=1}^G X_g'X_g \right)^{-1} 
      \sum_{g=1}^G  X_g' w_g^j Y_g, \\
\widehat{\sigma}_{b,j}^2 
&= \left(\frac{G}{b}\right)^2 
   r' \left(\sum_{g=1}^G X_g'X_g \right)^{-1} 
      \left( \sum_{g=1}^G \widehat{S}_{g,j} w_g^j \widehat{S}_{g,j}' \right) 
   \left(\sum_{g=1}^G X_g'X_g \right)^{-1} r,
\end{align*}
where 
$
\widehat{S}_{g,j} = X_g' \big( Y_g - X_g \widehat{\theta}_{b,j} \big).
$

Note that the inverse factor $\left(\sum_{g=1}^G X_g'X_g \right)^{-1}$ is computed from the full unweighted sample, whereas the linear component and its variance are constructed from the bootstrap sample. 
Practical motivations for this feature will be discussed in Remark~\ref{rem:non-invertibility}.

Define the bootstrapped empirical distribution function $\widehat{L}_{G,b}$ of 
$(\widehat{\delta}_{b,j} - \widehat{\delta})/\widehat{\sigma}_{b,j}$ by
\begin{align*}
\widehat{L}_{G,b}(t) 
= \frac{1}{M}\sum_{j=1}^M 
   \mathds{1}\!\left( \frac{\widehat{\delta}_{b,j} - \widehat{\delta}}
                            {\widehat{\sigma}_{b,j}} \le t \right).
\end{align*}
For any $a \in (0,1)$, define the corresponding critical value as
\begin{align*}
\widehat{c}_{G,b}(1-a) 
= \inf \left\{ t \in \mathbb{R} : \widehat{L}_{G,b}(t) \ge 1-a \right\}.
\end{align*}
As will be formally established in Section \ref{sec:theory}, this critical value is guaranteed to satisfy
\begin{align*}
P\!\left( \frac{\widehat{\delta} - \delta}{\widehat{\sigma}} 
   \le \widehat{c}_{G,b}(1-a) \right) \;\to\; 1-a
\end{align*}
as $G \to \infty$. 
Hence, the CS bootstrap method provides asymptotically valid inference.

\bigskip\noindent
{\bf Practical Implication:}
For the $t$-statistic, one may continue to use the conventional CR ``standard error'' 
$\widehat{\sigma}$ even though it may diverge.\footnote{Note that the ``standard error'' 
$\widehat{\sigma}$ does not converge in probability when $\alpha < 2$.} 
However, rather than relying on conventional Gaussian critical values 
(e.g., $\Phi^{-1}(0.025) \approx -1.96$ and $\Phi^{-1}(0.975) \approx 1.96$), 
one should instead employ the bootstrap-based critical values 
$\widehat{c}_{G,b}(0.025)$ and $\widehat{c}_{G,b}(0.975)$ obtained from the 
CS bootstrap procedure. 
These critical values, for example, can be used to construct a 95\% confidence interval for $\delta$. $\blacktriangle$
\medskip

\begin{remark}[Practicality of the method]	\label{rem:practicality}
Even though the convergence rate of $\widehat{\delta} - \delta$ is unknown, our inference remains valid because it is based on a self-normalized statistic. 
Moreover, our procedure does \emph{not} require estimation of the unknown stability index $\alpha$, 
nor is it necessary to estimate the slowly varying functions $L_1$ and $L_2$. 
These features represent important practical advantages of our proposed method, as these nuisance parameter estimation problems are well known to be challenging in the statistics literature.
$\blacktriangle$
\end{remark}

\begin{remark}[Finite sample non-invertibility of other resampling methods]\label{rem:non-invertibility}
Compared with conventional resampling methods, the CS bootstrap offers two advantages. 
First, because it does not require recomputation of the inverse factor at each bootstrap iteration, the method is computationally more efficient. 
Second, and more importantly, in finite samples, when the regressors include a cluster-specific binary treatment variable or other dummies that are highly correlated within a cluster, the matrix $\sum_{g = 1}^G w_g^j X_g'X_g$ is often singular for small $b$, as is common in cluster-RCT settings. 
Consequently, the resampled OLS estimator may be undefined in a non-negligible fraction of bootstrap iterations. 
This problem also arises in other cluster-based resampling methods, such as the jackknife, subsampling, and the conventional bootstrap. 
In practice, several \textit{ad hoc} ``fixes,'' such as employing a generalized inverse or dropping such realizations, are often used, though their theoretical justification remains unclear. 
By contrast, the proposed CS bootstrap procedure, which relies on the full unweighted matrix $\sum_{g = 1}^G X_g'X_g$, avoids this issue in a theoretically supported manner.
$\blacktriangle$
\end{remark}

\begin{remark}[Inference using parametric bootstrap]\label{rem:inference_using_asymptotic_distribution}
The $t$-statistic has a complicated but well-defined class of limiting distributions, 
as illustrated in Figure~\ref{figure:limit_distribution}. 
A natural alternative approach is to bootstrap critical values from this known limiting distribution for inference, as suggested in \cite{cornea2015parametric}. 
However, this requires estimation of the unknown parameters: the index of stability $\alpha$ 
and the measure of limiting symmetry $p$ as defined in \eqref{eq:tail_balancing}. 
Our simulation results show that inference based on estimated values of $\alpha$ and $p$ performs poorly in finite sample. 
Moreover, estimating $\alpha$ and $p$  requires selecting tuning parameters that are inherently \textit{ad hoc} choices. 
The resulting inference is highly sensitive to this tuning and remains imprecise unless the number of clusters is quite large (e.g., exceeding 2000). 
For these reasons, we do not recommend this parametric bootstrap-based approach in our setup.
$\blacktriangle$
\end{remark}

\begin{remark}[Subsampling]\label{rem:subsampling}
 If the i.i.d. count vectors $(w_1^j, \ldots, w_G^j)$ are instead generated by sampling $b$ out of $G$ units without replacement, the procedure would entail a version of score subsampling. Theoretical results for this alternative subsampling approach are established in a manner similar to those for the CS bootstrap. 
$\blacktriangle$
\end{remark}

\subsection{Theoretical Properties}\label{sec:theory}
We now provide theoretical support for our proposed CS bootstrap method. 
The following theorem provides a formal justification of its robust asymptotic validity. 

\begin{theorem}[Cluster-Robust Inference by the CS Bootstrap]\label{cor:thin_tail}
Suppose that Assumption \ref{assm} is satisfied.
If $b \to \infty$ and $b/G=o(1)$ as $G\to \infty$, and $M\to\infty$, then
\begin{align*}
\sup_{t\in \Real}|\hat L_{G,b}(t)-J^*(t)|\stackrel{p}{\to} 0
\end{align*}
for a continuous limiting distribution $J^*(\cdot)$. 
Thus, for any significance level $a \in (0,1)$,
\begin{align*}
P\left((\hat \delta - \delta)/\hat \sigma \le \hat c_{G,b}(1-a)\right)\to 1-a.
\end{align*}
\end{theorem}


The proof is non-trivial and proceeds by considering two distinct cases. The first, corresponding to $\alpha < 2$, is formalized in Lemma~\ref{theorem:main} in Appendix~\ref{sec:theory_alpha_less_than_two} and proved in detail in Appendix~\ref{sec:proof:theorem:main}. The second case, $\alpha = 2$, is analyzed in Appendix~\ref{sec:proof:cor:thin_tail}, which also synthesizes the two regimes to establish Theorem~\ref{cor:thin_tail}. Here, asymptotics are taken with respect to $G \to \infty$ for a given DGP; further studies on uniformity can be found in Appendix~\ref{sec:theory:uniform}.


While the theory requires $b \to \infty$ and $b/G = o(1)$ as $G \to \infty$, 
in practice the researcher must select a finite value of $b$. 
This choice should be neither too large nor too small. 
Intuitively, if $b$ is chosen too close to $G$, the largest clusters are sampled too frequently in the bootstrapped $t$-statistics, which prevents the procedure from adequately reflecting the heavy-tailed nature of the DGP when $\alpha < 2$.
Conversely, if $b$ is too small, the bootstrapped $t$-statistics become excessively noisy. 
Thus, one seeks a value of $b$ that lies in a stable range, 
such that small perturbations of $b$ (e.g., increasing or decreasing it by one) have only minimal impact on the bootstrap distribution. 
In the context of the $m$-out-of-$n$ bootstrap, \citet{bickel2008choice} formalized this idea and proposed a data-driven algorithm with theoretical guarantees for its validity. 
Here, we introduce a modified version of their algorithm tailored to our setting.


\begin{algorithm}[Data-Driven Choice of $b$]\label{alg:bickel_sakov}${}$
    \begin{enumerate}[(i)]
    \item Let $b_\ell=\lceil q^\ell \cdot \phi(G) \rceil$ for each $\ell \in \mathbb{N}$, where $\lceil a\rceil$ stands for the smallest integer greater than or equal to $ a$, $q \in (0,1)$, and $\phi$ is a strictly sub-linear function.\footnote{For instance, we set $\phi(G) = G^{0.99}$. This specification, together with $q = 0.99$, is used in our numerical studies and empirical application.}

    \item Obtain 
    $\hat L_{G,b_\ell}$
    for each $\ell \in \mathbb{N}$ by simulation.
    \item Set
    \[\hat b=\underset{\ell \in \mathbb{N}}{\argmin}\sup_{t\in \Real}\left|\hat L_{G,b_{\ell}}(t)-\hat L_{G,b_{\ell+1}}(t)\right|.\]
    	\item 	If there is a tie, let $\widehat{b}$ be the largest $b_\ell$ among them.
    \end{enumerate}
\end{algorithm}

The following theorem provides theoretical guarantees for the data-driven choice of $\widehat{b}$ in the CS bootstrap. 
Specifically, it shows that Theorem~\ref{cor:thin_tail} continues to hold with our data-driven choice of $\widehat{b}$. 

\begin{theorem}[Cluster-Robust Inference by the Data-Driven CS Bootstrap]\label{theorem:data-driven}
Suppose that Assumption \ref{assm} is satisfied and $b=\hat b$ is chosen according to Algorithm \ref{alg:bickel_sakov}. 
Then, the conclusion of Theorem \ref{cor:thin_tail} continues to hold.
\end{theorem}

\noindent
A proof is found in Appendix \ref{sec:theorem:data-driven}. Hence, following the selection of $\hat b$, asymptotically valid inference can be carried out using the corresponding critical value $\hat c_{G,\hat b}(1-a)$.

\section{Simulation Studies}\label{sec:simulation}

In this section, we present simulation studies evaluating the finite-sample performance of our proposed method of genuinely robust CR inference, based on the cluster score (CS) bootstrap, in comparison with conventional CR methods.

The data-generating design is defined as follows. 
We consider the cluster treatment model with individual covariates
$$
Y_{gi} = \theta_0 + \theta_1 T_g + \sum_{j=1}^K \theta_j X_{g,i,j+1} + U_{gi},
$$
following \citet*[][Equation~(40)]{MaNiWe2022fast}, among others. 
The binary treatment variable $T_g$ equals one for $\lceil 0.2G \rceil$ clusters and zero for the remaining $G - \lceil 0.2G \rceil$ clusters, where $\lceil a \rceil$ denotes the smallest integer greater than or equal to $a$. 
Cluster sizes are drawn independently as 
$N_g \sim \lceil \text{Pareto}(1,\alpha) \rceil$ for $g \in \{1,\ldots,G\}$. 
For each $g \in \{1,\ldots,G\}$, we independently draw $N_g$-variate random vectors 
$(\tilde X_{g1j},\ldots,\tilde X_{gN_gj})' \sim \mathcal{N}(0,\Omega)$ for $j \in \{1,\ldots,K\}$ 
and $(\tilde U_{g1},\ldots,\tilde U_{gN_g})' \sim \mathcal{N}(0,\Omega)$ in the baseline design, 
where $\Omega$ is an $N_g \times N_g$ covariance matrix with $\Omega_{ii}=1$ for all $i \in \{1,\ldots,N_g\}$ and $\Omega_{ii'}=1/2$ whenever $i \neq i'$. 
The controls are constructed as 
$X_{gij} = 0.2 F_{\text{Beta}(2,2)}^{-1} \circ \Phi(\tilde X_{gij})$, 
where $F_{\text{Beta}(2,2)}$ and $\Phi$ denote the CDFs of the $\text{Beta}(2,2)$ and standard normal distributions, respectively. 
The errors are constructed heteroskedastically as $U_{gi} = 0.2 \tilde U_{gi}$ if $T_g=0$ and $U_{gi} = \tilde U_{gi}$ if $T_g=1$.

We vary the exponent parameter $\alpha \in \{1.1,1.2,\ldots,1.9,2.0\}$ across simulation sets. 
The regression coefficients are fixed at $(\theta_0,\theta_1,\theta_2,\ldots,\theta_{K+1})' = (1,1,1,\ldots,1)'$ throughout, while the covariate dimension varies as $K \in \{5,10\}$. 
The sample size (i.e., the number of clusters) is set to $G = 50$ across all simulations, which is roughly comparable to the number of U.S.\ states. 
Each simulation set consists of 10,000 Monte Carlo iterations.


\begin{figure}[t]
\centering
\includegraphics[width=0.36\textwidth]{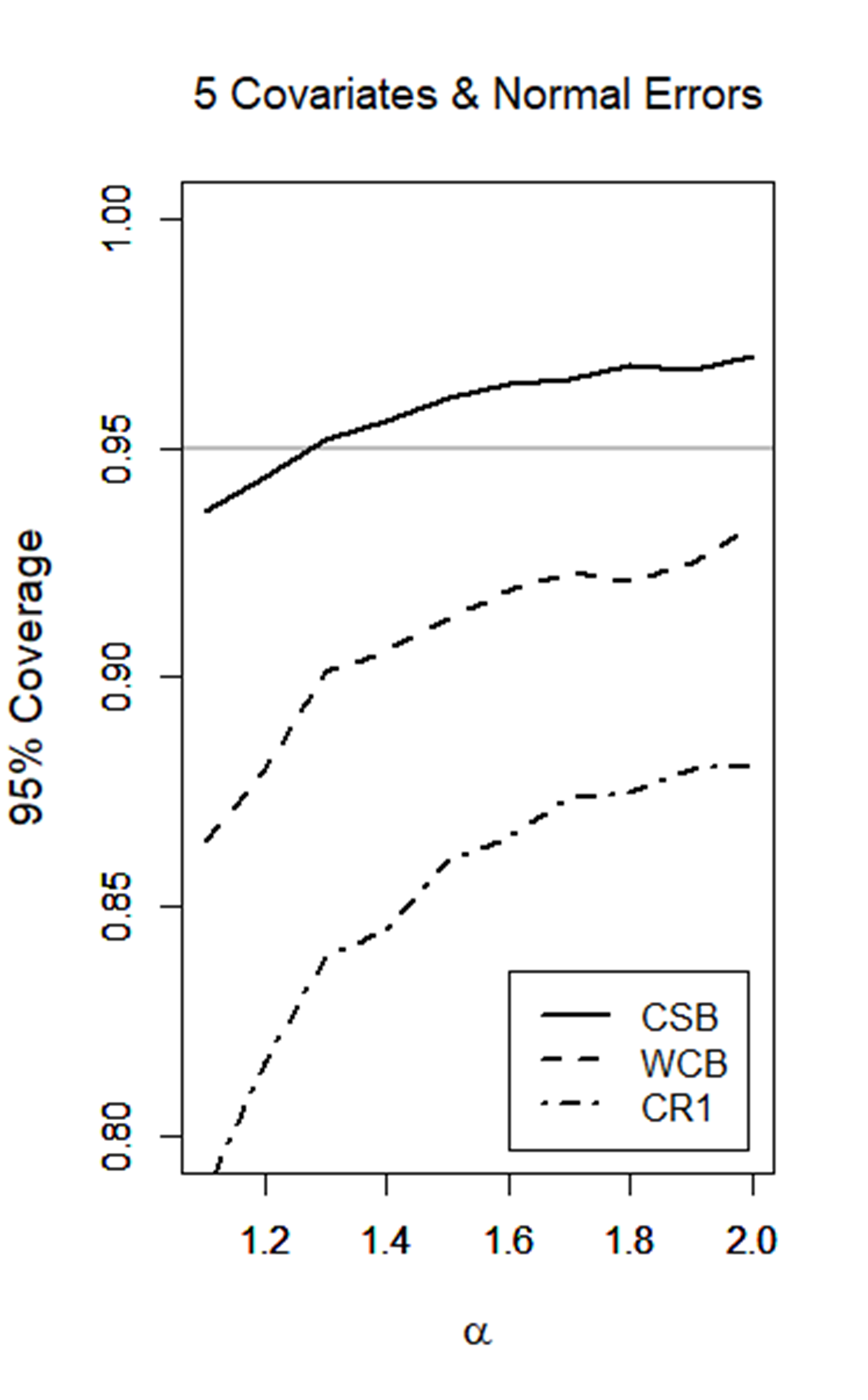}
\includegraphics[width=0.36\textwidth]{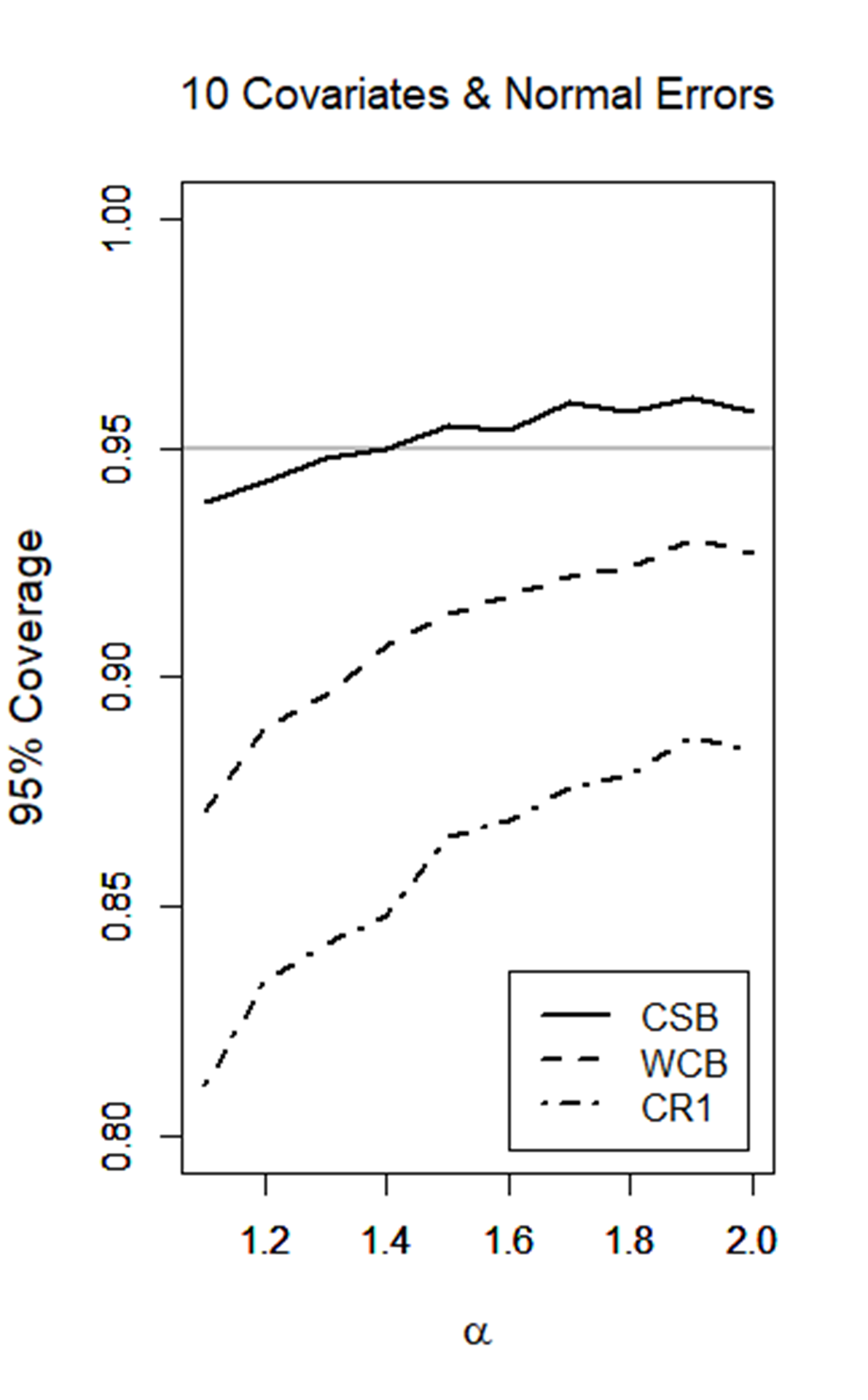}
\caption{Monte Carlo coverage frequencies for the baseline designs with normal errors. ``CSB'' (respectively,~``WCB'' and ``CR1'') denotes the CS bootstrap (respectively,~wild cluster bootstrap and CR1 standard errors with normal critical values). The nominal 95\% coverage probability is indicated by the horizontal gray line.}${}$
\label{fig:sim_normal}
\end{figure}

Figure~\ref{fig:sim_normal} reports the Monte Carlo coverage frequencies. 
The horizontal axis represents the value of $\alpha$, and the vertical axis represents the coverage frequency. 
In the legend, `CSB' denotes the CS bootstrap, while `WCB' and `CR1' denote the wild cluster bootstrap and the CR1 standard error with normal critical values, respectively. 
The nominal coverage probability of 95\% is indicated by the horizontal gray line at 0.95.

The CS bootstrap performs best, followed by the WCB and the CR1. 
Overall, the CS bootstrap consistently delivers coverage frequencies closest to the nominal 95\% level across the range of $\alpha$. 
By contrast, both conventional methods suffer from under-coverage, particularly for small values of $\alpha$.

\section{An Empirical Illustration}\label{sec:application}

\citet{akhtari2022political} study the effects of political turnover on various outcomes measuring the quality of public services in Brazil. 
In their original paper (Table~3, Column~5), they estimate the following linear model by OLS:
\begin{align*}
\text{score}_{gi+1} = \theta_0 + \theta_1 \cdot \mathds{1}\{\text{IVM}_{g}<0\} + \theta_2 \cdot \text{IVM}_{g} + \theta_3 \cdot \mathds{1}\{\text{IVM}_{g}<0\} \cdot \text{IVM}_{g} + \theta_4 \cdot \text{score}_{gi} + U_{gi}.
\end{align*}
The dependent variable, $\text{score}_{gi+1}$, is the test score of fourth-grade students in the year following an election. 
The main explanatory variable, $\text{IVM}_{g}$, is the incumbent vote margin. 
Thus, $\mathds{1}\{\text{IVM}_{g}<0\}$ equals one when the incumbent party loses the election. 
The parameter of interest is $\delta = \theta_1$, which measures the effect of political turnover on test scores.\footnote{This effectively implements a sharp regression discontinuity design, although the original paper estimates the effect by OLS using this linear specification.} 
While the original paper considers alternative `bandwidths,' we focus on the bandwidth 0.110 to maximize the sample size, following prior work \citep[][Sec.~7.2]{MaNiWe2022fast}, which replicates this regression.

The original paper clusters standard errors at the municipality level, and we follow this definition of the cluster unit. 
There are $G = 2101$ municipalities in the data, with $\max_{1 \leq g \leq G} N_g^2 / N \approx 26$. 
Thus, the assumption $\max_{1 \leq g \leq G} N_g^2 / N \to 0$, under which conventional CR inference methods are guaranteed to work, is difficult to justify in this application.

Table~\ref{tab:application} reports the $p$-values for $\delta = \theta_1$ based on alternative inference methods. 
Column HC1 reports the $p$-value using conventional inference without clustering, i.e., the HC1 standard error. 
Columns CR1 and WCB report the $p$-values using conventional CR inference methods, namely the CR1 standard error and the wild cluster bootstrap,\footnote{While there are four variants of WCB \citep[cf.][Sec.~7.2]{MaNiWe2022fast}, they yield identical $p$-values up to the reported digits, and hence we summarize them in a single column.} respectively, with normal approximation. 
Finally, column CSB reports the $p$-value based on our proposed inference method using the CS bootstrap. 
We employ the same code as in the simulation studies in Section~\ref{sec:simulation}, including the choice of $b$ based on the minimum volatility method.

\begin{table}
\centering
\begin{tabular}{cccccc}
\hline\hline
& HC1 & CR1 & WCB & CSB\\
\hline
p-value & 0.000 & 0.006 & 0.006 & 0.266\\
\hline\hline
\end{tabular}
\caption{$p$-values for the effect $\delta = \theta_1$ of political turnover on fourth-grade students' test scores, based on conventional inference methods (HC1, CR1, WCR) and our proposed CS bootstrap (CSB).}${}$
\label{tab:application}
\end{table}

The $p$-value is zero up to the third digit when standard errors are not clustered (HC1). 
Conventional CR inference methods (CR1 and WCB) with normal approximation yield larger $p$-values, but the statistical significance remains unchanged. 
By contrast, our proposed CS bootstrap method produces a much larger $p$-value, rendering the effect $\delta = \theta_1$ statistically insignificant, unlike any of the conventional methods. 
These results highlight that failing to account for potential non-Gaussianity in the limiting distributions, particularly in the presence of unignorably large clusters, can lead to erroneous statistical conclusions.

\section{Summary}

Conventional methods for cluster-robust inference often fail to provide consistent results in the presence of unignorably large clusters. 
In this paper, we formalize this limitation by deriving a necessary and sufficient condition for consistency. 
We document that 77\% of empirical research articles published in the \textit{American Economic Review} and \textit{Econometrica} during 2020--2021 contain model specifications fail to satisfy this condition. 

To address this challenge, we propose the CS bootstrap and establish its size control across a wide class of data-generating processes where conventional methods break down. 
Our simulation studies confirm the reliability and effectiveness of the proposed method, underscoring its practical value in overcoming the limitations of existing cluster-robust inference techniques. We further demonstrate the failure of the wild cluster bootstrap in Section \ref{sec:inconsistent_wcb} and discuss the related uniformity issues in Section \ref{sec:theory:uniform} in the appendix, reinforcing the need for our proposed approach.
Finally, we demonstrate that correctly accounting for potential non-Gaussianity can overturn empirical conclusions. We conclude the paper by modifying a well-known \textit{haiku} by \cite{hiranohaiku}.


\begin{quote}
T-stat looks too good.\\
Use \textit{the cluster score bootstrap}--\\
significance gone.
\end{quote}

\newpage
\section*{Appendix}\appendix
\section{Omitted Details}\label{sec:details}

This appendix section collects technical details that are omitted from the main text.

\subsection{Alternative Characterization of $\xi_g$ Belonging to the Domain of Attraction of an $\alpha$-Stable Distribution for $\alpha<2$}\label{sec:alternative}

Citing a result from the existing literature, this section presents complete details about the power law characterization \eqref{eq:power} discussed in Section \ref{sec:model} in the main text.

\begin{theorem}[\citealp{pena2009self}, Theorem 2.24]\label{theorem:pena}
If $\alpha<2$,
then $\xi_g$ belongs to the domain of attraction of an $\alpha $-stable distribution if and only if  
\begin{align}
&P(|\xi_g|> t)= t^{-\alpha}L(t) \qquad\text{and} \label{eq:pena1}\\
&\lim_{t\to \infty}\frac{P(\xi_g> t)}{P(|\xi_g|> t)}=p \label{eq:pena2}
\end{align}
for some $p \in [0,1]$ and some slowly varying function $L(\cdot)$.
\end{theorem}

The first condition \eqref{eq:pena1} means that the tail limit of the absolute value of the random variable of interest  has an approximately Pareto tail, or so-called power law.
Known as the balancing condition, the second condition \eqref{eq:pena2} in this alternative characterization imposes a mild restriction on the existence of limiting ratios of one-sided tail probabilities over the two-sided tail probability; it rules out some pathological, infinitely oscillating type situations such that these limiting ratios do not exist. This condition only imposes restrictions in the limit and accommodates a wide range of tail behaviors as $p$ are permitted to be even $0$ and $1$, thereby allowing cases from distributions that are bounded on one side to distributions with heavy two-sided tails.  

\subsection{Auxiliary Theory Focusing on Cases with $\alpha<2$}\label{sec:theory_alpha_less_than_two}

This section presents a lemma that we state and prove on the way to proving Theorem \ref{cor:thin_tail} in Section \ref{sec:theory} in the main text.
Namely, for ease of writing, we state our main result focusing on cases with $\alpha \in (1,2)$.
An extension of this result to the general cases with $\alpha \in (1,2]$ follows as Theorem \ref{cor:thin_tail} with additionally accounting for the case with $\alpha=2$.

\begin{lemma}\label{theorem:main}
Suppose that Assumption \ref{assm} is satisfied for $\alpha \in (1,2)$.
If $b\to \infty$, $b/G=o(1)$, and $M\to\infty$ as $G\to \infty$, then
\begin{align*}
\sup_{t\in \Real}|\hat L_{G,b}(t)-J^*(t)|\stackrel{p}{\to} 0,
\end{align*}
and thus
\begin{align*}
P\left((\hat \delta - \delta)/\hat \sigma \le \hat c_{G,b}(1-a)\right)\to 1-a.
\end{align*}
\end{lemma}

A proof is provided in Appendix \ref{sec:proof:theorem:main}.

\begin{remark}[Heavy-tailed cluster sums]
In this lemma, we essentially assume that the tails of the distributions of $\|S_g\|$ and $\|X_g'X_g\|$ both follow the power law with the shape parameter (Pareto exponent) in $(1,2)$, which implies that the variances of $S_g$ and $(X_g' X_g) $ do not exist. 
See Appendix \ref{sec:alternative}.
This is a rather general condition in the sense that the heavy tail can come from the distribution of cluster sizes $N_g$, the distribution of individuals' $(X_{gi}',U_{gi})$, or both. 
$\blacktriangle$
\end{remark}

\begin{remark}[Unignorability and impossibility of normal approximation]\label{remark:impossibility}
An inspection of the proof of Lemma \ref{theorem:main}, combined with Remark 2 in \cite{lepage1981convergence}, unveils that, when $\alpha<2$, the tails of the first component of representation (\ref{eq:rates_1}) satisfies  
\begin{align*}
P\left(\left|\epsilon_1Z_1-(2p-1)\E[Z_1\1(Z_1<1)]\right|>t\right)\sim P\left( \left|\sum_{k=1}^\infty\{\epsilon_kZ_k-(2p-1)\E[Z_k\1(Z_k<1)]\}\right|>t\right)
\end{align*}
as $t\to\infty$.
Since the term $\left|\epsilon_1Z_1-(2p-1)\E[Z_1\1(Z_1<1)]\right|$  corresponds to the limiting distribution of the absolute value of the scaled score of the largest cluster, it has an asymptotically unignorable influence on the limiting $\alpha$-stable distribution -- see also Section 1.4 in \citet{samorodnitsky1994stable}. 
For ease of illustration, suppose that the regressor and error distributions are uniformly bounded and $\Cov(X_{gi} U_{gi}, X_{gi}U_{gi'}|N_g)\ge \underline c >0$  for all $i=1,...,N_g$ with probability one. 
This then implies
\begin{align*}
\frac{\max_{g=1,...,G} \|S_g\|}{N}\sim_p\frac{\max_{g=1,...,G} N_g}{N}\gg 0,
\end{align*} 
which directly violates the necessary and sufficient condition for the asymptotic variance to be estimable derived in Corollary 4.1 in \cite{kojevnikov2021some}, as well as the
 conventional assumption
\begin{align*}
\frac{\max_{g=1,...,G} N_g^2}{N}=o_p(1),
\end{align*} 
required in the literature (e.g. Assumption 2 in \citealt{HaLe19}) for normal approximation.\footnote{It is assumed in the literature of CR inference based on the normal approximation that
$
\frac{\max_{g=1,...,G} N_g^2}{N}=o_p(1).
$
When $\mathbb{E}[N_g]=c>0$ exists, this assumption is equivalent to
$
\frac{\max_{g=1,...,G} N_g^2}{G}=o_p(1). 
$}

In addition,  a necessary and sufficient condition for the limiting distribution of sums of independent random variables to be normal is the uniform asymptotic negligibility condition, i.e., the largest summand in absolute value has an asymptotically negligible contribution to the sum \citep[cf.][Theorem 23.13]{davidson1994stochastic}. 
Thus, it is impossible to derive a theoretically valid normal-approximation-based procedure of inference in the presence of unignorably large clusters without imposing restrictions on within-cluster dependence.
$\blacktriangle$
\end{remark}

\begin{remark}[On CR standard error estimation]
The test statistic we consider is the standard t-statistic used in the literature. 
Its denominator consists of a CR standard error without imposing a null hypothesis. 
When $\alpha<2$, the asymptotic variance does not exist, and nor is this ``standard error'' consistent but remains random asymptotically. 
This is similar in spirit to the fixed-$b$ asymptotics \citep[e.g.,][]{kiefer2002heteroskedasticity} in the literature of long-run variance estimation, although the underlying theory is completely different as the fixed-$b$ asymptotics crucially relies on normal approximation and the functional central limit theorem. 
Showing that this ``standard error'' with estimated residuals has negligible impact on the asymptotic distribution requires a completely different proof strategy from the conventional approach of those taken in the proof of Theorem 7.6 in \cite{hansen2022econometrics}.
$\blacktriangle$
\end{remark}

\subsection{Inconsistency of the Wild Cluster Bootstrap under $\alpha<2$.}\label{sec:inconsistent_wcb}

The wild cluster bootstrap \citep{cameron2008bootstrap} is a popular alternative resampling method of CR inference.
It has been shown in various simulation studies to behave well under $\alpha=2$.
Validity of the wild cluster bootstrap in cases of $\alpha=2$ has been shown in \citet{DjMaNi19} under fairly general conditions. 
As their proof relies crucially on Lyapunov's CLT, however, their arguments do not hold under $\alpha<2$ -- see Remark \ref{remark:impossibility}. 
A remaining and potentially more interesting question is whether one can prove its validity using an alternative argument.
The following result suggests that such efforts are ill-fated when $\alpha<2$.

For simplicity of illustration, consider the case of a univariate regression with only the intercept, i.e. a cluster sampled mean
$
\hat \theta = N^{-1}\sumg \sumi Y_{gi}
$
with the cluster specific population mean normalized to $\theta=\E\left[\sumi Y_{gi}\right]=0$ without loss of generality.
Suppose that the parameter of inference is $\theta$.
Under the null hypothesis $H_0:\theta=0$, the standard CR t-statistic can be formed as
\begin{align*}
T_{G}=\frac{\sumg \sumi Y_{gi}}{\sqrt{\sumg \left(\sumi (Y_{gi}-\hat \theta)\right)^2}}.
\end{align*}
The wild-cluster-bootstrap version of the estimator is defined by $\hat\theta^*=N^{-1}\sumg v_g^*\sumi Y_{gi}$, where $(v_g^*)_{g=1}^G$ are i.i.d. Rademacher auxiliary random variables generated by a researcher independently from the observed data $\{\{Y_{gi}\}_{i=1}^{N_g}\}_{g=1}^G$. 
The null-imposed wild cluster bootstrap test statistic is defined by
\begin{align*}
T_{G}^{*}=\frac{\sumg v_g^*\sumi Y_{gi}}{\sqrt{\sumg \left(v_g^*\sumi (Y_{gi}-\hat \theta^*) \right)^2}}.
\end{align*}
We introduce the notation $Y_{1:G}=(Y_{gi}:g=1,...,G, i=1,...,N_g)$ for convenience.
As Lemma \ref{theorem:main} implies continuity of the limiting distribution of $T_G$, 
the wild cluster bootstrap is consistent if
\begin{align*}
\sup_{t\in \Real}\left|P(T_G^*\le t|Y_{1:G})-P (T_G\le t)\right|=o_p(1)
\qquad\text{as }G\to \infty.
\end{align*}

\begin{theorem}[Inconsistency of the wild cluster bootstrap]\label{theorem:bootstrap}
	Under the above setup and Assumption \ref{assm}, if $\alpha \in (1,2)$, then the wild cluster bootstrap with Rademacher auxiliary random variables is inconsistent. 
\end{theorem}

A proof can be found in Appendix \ref{sec:proof:theorem:bootstrap}. Note that the proof strategy for the numerator of the test statistic is closely related to the approach taken in \cite{cavaliere2013wild}, which demonstrates that a suitably modified wild bootstrap can be valid for symmetrically distributed i.i.d. sample means with infinite variance.

\subsection{Details of Section \ref{sec:fragility:heuristic}}\label{sec:appendix:self_normalized}

Section \ref{sec:fragility:heuristic} argues that the self-normalized CLT may or may not hold under our framework.
This appendix section presents details of this argument. 

For the estimand $\theta = \mathbb{E}[Y_{gi}]$ for simplicity, consider the estimator
\begin{equation*}
\hat{\theta}=\frac{1}{N}\sum_{g=1}^{G}S_{g},
\end{equation*}
where $S_{g}=\sum_{i=1}^{N_{g}}Y_{gi}$ and $N=\sum_{g=1}^{G}N_{g}$. 
For simplicity, assume that $Y_{gi}$ is identically distributed with mean zero and variance one, and that $Y_{gi}$ is independent from $N_g$.
Also, assume the cluster-sampling framework in which observations are independent across $g$. 
Let $\Omega _{N}$ denote the variance of $\sqrt{N}\hat{\theta}$, i.e., $\mathbb{E}[ N\hat{\theta}^{2}]$.

We now consider three cases of within-cluster dependence: (i) $Y_{gi}$ is i.i.d. across $i$ within each $g$ (i.e., no cluster dependence); (ii) $Y_{gi}=Y_{gj}$ for all $i$ and $j$ within the same cluster (i.e., the strongest form of cluster dependence); and (iii) a combination of the cases (i) and (ii).

\begin{description}
\item[Case (i)] 
Suppose that $Y_{gi}$ is i.i.d. across $i$.
The self-normalized CLT considers
\begin{equation*}
\left(\mathbb{E}[\hat{\theta}^{2}]\right)^{-1/2}\hat{\theta} \overset{d}{\rightarrow }\mathcal{N}\left( 0,1\right).
\end{equation*}
Since
$
\mathbb{E}[ N\hat{\theta}^{2} ] =\mathbb{E}\left[ Y_{gi}^{2}\right]=1
$
under the independence across $i$ and $g$, we have
\begin{equation}\label{eq:self_normalized}
\left(\mathbb{E}[\hat{\theta}^{2}]\right)^{-1/2}\hat{\theta}
=
\sqrt{N}\hat{\theta}=\frac{G^{-1/2}\sum_{g=1}^{G}S_{g}}{\sqrt{\frac{1}{G}\sum_{g=1}^{G}N_{g}}}.
\end{equation}
By the law of large numbers and the assumption that $N_{g}$ is regularly varying with exponent $\alpha >1$, we have
\begin{equation*}
\frac{1}{G}\sum_{g=1}^{G}N_{g}\overset{d}{\rightarrow }\mathbb{E}\left[ N_{g}
\right] <\infty 
\end{equation*}
for the denominator of \eqref{eq:self_normalized}.
The independence within cluster implies that conditional on $\{N_g\}^G_{g=1}$, 
\begin{equation*}
\frac{1}{\sqrt{N}}\sum_{g=1}^{G}\sum_{i=1}^{N_{g}}Y_{gi}\overset{d}{\rightarrow }\mathcal{N}\left( 0,1\right) \text{.}
\end{equation*}
for the numerator of \eqref{eq:self_normalized}.
Therefore, the self-normalized CLT still holds, but with the convergence rate being $N^{-1/2}$, instead of $G^{-1/2}$ if we treat $\{N_g\}^G_{g=1}$ as fixed sequences of constants. 
Now consider $\{N_g\}^G_{g=1}$ as random variables. 
Given the Pareto tail of $N_{g}$, we have that
\begin{equation*}
N=\sum_{g=1}^{G}N_{g} =O_{p}\left( G\right).
\end{equation*}
It follows that $G^{-1/2}\sum_{g=1}^{G}\sum_{i=1}^{N_{g}}Y_{gi}=O_{p}\left(1\right) $. 

\item[Case (ii)] Consider the case with perfect within-cluster dependence, i.e., $Y_{gi}\equiv Y_{g}$ for all $i\in \{1,...,N_{g}\}$ for each $g$.
In this case, $S_{g}=\sum_{i=1}^{N_{g}}Y_{gi}=N_{g}Y_{g}$, yielding that
\begin{equation*}
\sqrt{N}\hat{\theta}=\frac{G^{-1/2}\sum_{g=1}^{G}N_{g}Y_{g}}{\sqrt{\frac{1}{G}\sum_{g=1}^{G}N_{g}}}.
\end{equation*}%
The denominator still converges to $\sqrt{\mathbb{E}\left[ N_{g}\right] }$.
For the numerator, since $N_{g}Y_{g}$ is i.i.d. across $g$ and the two factors are independent with regularly varying tails, \citet[][Proposition 1.3.9]{mikosch1999regular} implies that the product $N_{g}Y_{g}$ also has regularly varying tail with exponent $\alpha<2$. 
Therefore, $G^{-1/2}\sum_{g=1}^{G}Z_{g} = G^{-1/2}\sum_{g=1}^{G}N_{g} Y_{g}$ is no longer $O_{p}\left( 1\right) $. 
More specifically, $\text{Var}[ G^{-1/2}\sum_{g=1}^{G}N_{g} Y_{g}]$ is equal to $\text{Var}[N_g]\cdot \text{Var}[Y_g]$, which is infinite given $\alpha<2$. 
In fact, \citet[][Theorem 1]{Geluk2000} implies that if the distribution of $N_gY_g$ is $\alpha$-stable, under some sequences of constants $a_G \simeq n^{1/\alpha} \rightarrow \infty$ and $b_G \in \mathbb{R}$, the limiting distribution
\begin{equation*}
\lim_{G \rightarrow \infty} P\left( \frac{1}{a_{G}} \sum_{g=1}^G S_g - b_G > x \right)
\end{equation*}
has the characteristic function
\begin{equation*}
\psi _{\alpha}\left( s\right) 
=
\exp \left\{ -\left( \left\vert s\right\vert ^{\alpha}+is \left( 1-\alpha\right) \tan(\alpha\pi/2) \frac{\left\vert s\right\vert ^{\alpha -1}-1}{\alpha -1}\right) \right\}. 
\end{equation*}
Thus, the CLT fails, and the asymptotic distribution will be non-Gaussian.
Therefore, even the jackknife standard error fails in this scenario. 
See, for example, Figures 5 and 6 in \citet{MaNiWe2022fast}.

\item[Case (iii)] Combining the above two cases, we now consider 
\begin{equation*}
Y_{gi}=\rho _{G}R_{g}+U_{gi},
\end{equation*}
where $R_{g}$ can be thought as a cluster-specific random effect and $U_{gi}$ is a random noise, which is i.i.d. across both $i$ and $g$. 
The normalizing constant $\rho _{G}$ determines the weights of $R_{g}$ in $Y_{gi}$. 
Under this setting, we have
\begin{align}
\sqrt{N}\hat{\theta} =&\frac{G^{-1/2}\sum_{g=1}^{G}S_{g}}{\sqrt{\frac{1}{G}\sum_{g=1}^{G}N_{g}}} \notag
\\
=&\frac{G^{-1/2}\rho _{G}\sum_{g=1}^{G}N_{g}R_{g}}{\sqrt{\frac{1}{G} \sum_{g=1}^{G}N_{g}}}+\frac{G^{-1/2}\sum_{g=1}^{G}\sum_{i=1}^{N_{g}}U_{gi}}{\sqrt{\frac{1}{G}\sum_{g=1}^{G}N_{g}}}. \label{eq:two part}
\end{align}
Following the same arguments as those in Case (ii), the first item above is asymptotically non-Gaussian (after some suitable normalization), but the second term is asymptotically normal. 
The orders of magnitudes of them depend on the distribution of $\left( R_{g},N_{g},U_{gi}\right) $. 
For example, if $\mathbb{E}\left[ R_{g} \right] =0 $ and $\mathbb{E}\left[ R_{g}^{2}\right] <\infty $, then $N_gR_g$ again has a regularly varying tail with exponent $\alpha<2$ \citep[e.g.,][Theorem 3]{EmGo80}. 
The generalized central limit theorem \citep[e.g.,][Chapter 11]{Ib13} implies that
$
\sum_{g=1}^{G}N_{g}R_{g} \simeq_p G^{1/\alpha }.
$
For the second term in \eqref{eq:two part}, Case (i) derives that $G^{-1/2}\sum_{g=1}^{G}\sum_{i=1}^{N_{g}}U_{gi}=O_{p}\left(1\right) $. 
The non-Gaussian part then dominates the normal part if $\rho_G G^{1/\alpha-1/2}\rightarrow\infty$ as $G\rightarrow\infty$. 
Since $\alpha<2$, a constant $\rho_G$ will satisfy this condition. 
\end{description}

As a final remark, we note that Assumption 3 in \citet{DjMaNi19} could relax the condition on $N_g$ into that $\max_{1\le g\le G} N_g/N \rightarrow 0$ when the within-cluster dependence is strong. 
The stochastic counterpart of this assumption fails under our framework where $N_g$ is treated as a random variable. 
More specifically, consider Case (ii) again for illustration. 
Let $\mu_N$ denote the reciprocal of the variance of $\hat{\theta}$ conditional on $\{N_g\}^G_{g=1}$ as in \citet{DjMaNi19}. 
The above derivation yields 
\begin{align*}
\text{Var}[\hat{\theta}|\{N_g\}^G_{g=1}] 
= &\, \frac{\sum^G_{g=1}N_g^2\text{Var}[Y_{gi}]}{(\sum^G_{g=1}N_g)^2} \\
= &\, \frac{\sum^G_{g=1}N_g^2\text{Var}[Y_{gi}]G^{-2}}{(G^{-1}\sum^G_{g=1}N_g)^2} \\
\simeq_p &\, G^{2/\alpha-2},
\end{align*}
and hence $\mu_N \simeq_p G^{2-2/\alpha}$. 
Therefore, for any constant $\lambda>0$, we have 
\[
\mu_N^{\frac{2+\lambda}{2+2\lambda}}\frac{\sup_g N_g}{N} \simeq_p G^{\rho(\lambda)},
\]
where $\rho(\lambda) = (2-2/\alpha)[(2+\lambda)/(2+2\lambda)-1/2]$. 
Recall $\alpha \in (1,2)$, yielding that $\rho(\lambda)>0$ for all $\lambda>0$.
Then, the above term diverges with probability approaching one. 

\section{Discussions on Uniformity}\label{sec:theory:uniform}

In the main text, all asymptotic properties are derived under a fixed data-generating process (DGP).
This appendix extends the analysis to uniform size control over a class of DGPs.
Without uniformity, for any given $G$ (regardless of its magnitude), there may exist a DGP $P_G$ such that the rejection probability under the null fails to converge to the nominal level.
Uniformity therefore guarantees more reliable inference in finite samples, particularly when $G$ is moderate.
To simplify the notations and assumptions, we focus on inference for the mean of a scalar random variable in the current subsection.

Consider a triangular array setup: for each $G\in \mathbb N$, suppose that we have an i.i.d. sequence $(S_g)_{g=1}^G=(S_{g,G})_{g=1}^G$, whose distribution is now $P = P_G$. 
Recall that  
\begin{align*}
(\hat \delta -\delta)=\frac{1}{G} \sum_{g=1}^G  S_g \quad\text{ and }\quad \hat\sigma^2= \frac{1}{G}\sumg \hat S_g^2,
\end{align*}
where $\hat S_g= S_g- G^{-1}\sumg S_g$. The test statistic of interest is again the t-ratio $(\hat \delta - \delta)/\hat\sigma$.
Henceforth, we will let $\E_P[\cdot]$ denote the expectation with respect to the DGP, $P$, if we are to emphasize such a dependence.
For any  $\varepsilon\in[0,1)$,
define $\mathbf P_1(\varepsilon)$ as the set of all the DGPs, $P$, such that there exist some $p\in [0,1]$ and $\alpha\in [1+\varepsilon,2)$ such that
\begin{align}
&\lim_{t\to \infty}\frac{P(S_g>t)}{P(|S_g|>t)}=p, \quad\text{ and }\label{cond:balanced_tails}\\
&P(|S_g|>t)=t^{-\alpha} L_P(t)\quad \text{as } t\to\infty\label{cond:regular_varying}
\end{align}
for an $L_P(\cdot)$ slowly varying at $\infty$ that can depend on $P = P_G$, and $\E_P[S_g]=0$ when it is defined.
In addition, define $\mathbf P_2$ as the set of all DGPs satisfying $\E_P[S_g]=0$ and the following uniform integrability condition
\begin{align*}
&\lim_{\lambda\to \infty}\sup_{P\in \mathbf P_2} \E_P\left[\frac{|S_g-\E_P[S_g]|^2}{\sigma^2(P)}\1\left\{\frac{|S_g-\E_P[S_g]|}{\sigma(P)}>\lambda\right\}\right]=0,
\end{align*}
where $\sigma^2(P)=\E_P[S_g^2]$ is finite.
Finally, define $\mathbf P(\varepsilon)=\mathbf P_1(\varepsilon)\cup\mathbf P_2$. 
The first set $\mathbf P_1(\varepsilon)$ covers the DGPs with heavy tail distributions and with regularly varying tail probabilities so that the variances of $S_g$ are infinite. 
The second set $\mathbf P_2$ covers a rich subset of DGPs in which the variances of $S_g$ are always finite and contains, in particular,  the set of DGPs with $2+\epsilon$ moments for any $\epsilon>0$.  It rules out certain examples such as those in the classical Bahadur-Savage example under which the $t$-test fails its size control for every sample size; see \cite{romano2004non} for more details.

First, we note that when $\alpha=1$, the t-ratio does not converge in distribution in general, except in very special situations. The following is a direct implication of \citet[p. 798]{logan1973limit}. 
\begin{proposition}\label{prop:Cauchy}
When $\alpha=1$ in \eqref{cond:regular_varying}, the t-ratio
$
(\hat \delta-\delta)/\hat \sigma
$ converges weakly to a nondegenerate limiting distribution only if $S_g$ belongs to either the domain of attraction of a Cauchy law or a translate of Cauchy law. Hence, no confidence set constructed using quantiles of the asymptotic distribution of the t-ratio can achieve uniform size control over $\mathbf{P}(0)$.
\end{proposition}

Nonetheless, we show a next best result holds true: 
 the proposed inference can control size uniformly over the set $\mathbf P(\varepsilon)$ if $\varepsilon>0$. 
  Recall that if the weight vectors \( (w_1^j,...,w_G^j) \) in the bootstrap are instead generated following sampling $b$ out of $G$ unit without placement, then the proposed procedure becomes subsampling.
Denote the empirical CDFs for the complete subsampling
\begin{align*}
L_G(x,P)=&\frac{1}{B_G}\sum_{j=1}^{B_G}\1\left\{\frac{\hat \delta_{b,j} - \delta}{\hat \sigma_{b,j}}\le x\right\},\quad
\hat L_G(x)=\frac{1}{B_G}\sum_{j=1}^{B_G}\1\left\{\frac{\hat \delta_{b,j} - \hat\delta}{\hat \sigma_{b,j}}\le x\right\},
\end{align*}
where $B_G={G \choose b}$.
Further, let the $a$-th quantile of  $\hat L_G(\cdot)$ be denoted by $\hat L_G^{-1}(a)$.
\begin{theorem}[Uniformity]\label{theorem:uniformity}
For any $\varepsilon\in(0,1]$,	the confidence sets constructed based on the proposed subsampling procedure achieves asymptotically uniform size control over $\mathbf P(\varepsilon)$. Explicitly, for any nonnegative $a_1$ and $a_2$ such that $0\le a_1+a_2<1$, we have
	\begin{align*}
	\lim_{G\to \infty}\inf_{P\in \mathbf P} P\left(\hat L_G^{-1} (a_1) \le \frac{\hat\delta-\delta}{\hat \sigma} \le \hat L_G^{-1} (1-a_2) \right)=1-a_1-a_2.
	\end{align*}	
    Furthermore, if $b^2/G=o(1)$, the same conclusion holds for CS bootstrap in place of subsampling.
\end{theorem}
A proof can be found in Appendix \ref{sec:proof:theorem:uniformity}.
The proof utilizes the general results in \cite{romano2012uniform} under high-level conditions together with our Lemma \ref{lemma:weak_conv_varying_alpha_p} in Appendix \ref{sec:proof:theorem:uniformity}. 
This new lemma establishes a novel convergence in distribution result for row-wise i.i.d. triangular arrays. 
Specifically, we consider the sequence of indices $\alpha_G \to \alpha_0 \in [1+\varepsilon,2]$ as $G\to \infty$, covering the cases with both normal ($\alpha_0 = 2$) and non-Gaussian ($\alpha_0 < 2$) limiting distributions. 
Recall that the t-test is not uniformly valid over the set of all DGPs with finite second moments, while it controls size uniformly over the set of all DGPs with finite $2+\epsilon$ moments for any $\epsilon>0$ (see e.g. \citealt{romano2004non}). 
Our result with $\mathbf{P}(\varepsilon)$ for all $\varepsilon>0$ is analogous to this classic result, although it extends the scope of uniformity to a much larger class of DGPs with potentially infinite second moments and non-Gaussian limiting distributions.

Furthermore, we note that it is likely that the additional condition $b^{2}/G=o(1)$ imposed for the CS bootstrap can be removed by appealing to, e.g. the generic results of \citet{andrews2020generic}. The present argument, for its simplicity, relies on exploiting the asymptotic equivalence between subsampling and the $m$-out-of-$n$ bootstrap under the regime $m^{2}/n =o(1)$, together with the high-level uniformity result for subsampling in \citet{romano2012uniform}, and is therefore not optimally tailored to the CS bootstrap setting. We leave a refinement of the proof along these lines for future research.

Finally, it is noteworthy that our uniform size control property exhibits resemblances to certain instances in the existing literature. 
An example is the unit root model presented in Example 1 of \cite{andrews2020generic}, where uniform size control persists across DGPs leading to either normal or non-Gaussian limiting distributions. 
In that example, \cite{andrews2020generic} demonstrate the continuity of their limiting distribution in a local parameter $h$ throughout its support, akin to the role served by our nuisance parameters $(\alpha, p)$ in our asymptotic theory. 
Notably, while infinite variance poses no hindrance in \cite{andrews2020generic}, its presence significantly complicates the analytical framework within our study. 
To the best of our knowledge, Theorem \ref{theorem:uniformity} stands as the first theoretical result addressing the uniformity property of subsampling or bootstrap for statistical models that may exhibit potentially infinite variance.

\section{Mathematical Proofs}\label{sec:proof}

This section collects all the mathematical proofs. 
The order in which the proofs appear differs from that of the corresponding statements in the main text for the following reasons. 
The proof of Theorem~\ref{cor:thin_tail} relies on Lemma~\ref{theorem:main}, and hence we present the proof of Lemma~\ref{theorem:main} before that of Theorem~\ref{cor:thin_tail}. 
Similarly, the proof of Theorem~\ref{thm:iff} depends on both Lemma~\ref{theorem:main} and Theorem~\ref{cor:thin_tail}, so we present those results before turning to Theorem~\ref{thm:iff}. 
Proofs for all remaining theorems are given in the order of appearance of their corresponding statements, namely, Theorem~\ref{theorem:bootstrap} (failure of the wild bootstrap) and Theorem~\ref{theorem:uniformity} (uniformity). 
Some additional technical lemmas are relegated to Appendix~\ref{sec:auxiliary_Lemmas}.

\subsection{Proof of Lemma \ref{theorem:main}}\label{sec:proof:theorem:main}

\begin{proof}[Proof of Lemma \ref{theorem:main}]
We establish the result for both the CS bootstrap and subsampling.
Without loss of generality, suppose that $X_{gi}$ is a scalar and $r=1$, and hence $\delta=\theta$.
The proof is divided into two steps.
In the first step, we derive the asymptotic distribution of the self-normalized sums that consist of the linear component of the influence function of the estimator. 
In the second step, we derive the validity of the proposed CS bootstrap.

\noindent \textbf{Step 1.}
Recall that 
\begin{align*}
\hat \theta -\theta=\left(\sum_{g=1}^G X_g'X_g \right)^{-1} \sum_{g=1}^G  S_g.
\end{align*}
We shall derive the asymptotic distribution for the following self-normalized sums of the linear component $\sumg S_g$:
\begin{align}
SN_{1G}(\theta):=\frac{\sum_{g=1}^G  S_g}{\sqrt{\sum_{g=1}^G  S_g^2}},\qquad
SN_{2G}(\theta):=\frac{\sum_{g=1}^G  S_g}{\sqrt{\sum_{g=1}^G  \hat S_g^2}},\label{eq:SN_sums}
\end{align}
where $\hat S_g= X_g'\hat U_g$.
The asymptotic distribution of a properly re-scaled $(\hat \theta - \theta)$ will then follow straightforwardly from the multiplication of $Q^{-1}$ on both the numerator and the denominator.
Since $\alpha \in(1,2)$, Corollary 1 in \cite{lepage1981convergence} yields
{ 
\begin{align}
SN_{1G}(\theta)\stackrel{d}{\to}  \frac{\sum_{k=1}^\infty\{\epsilon_k Z_k -(2p-1)\E[Z_k\1(Z_k < 1)]\} }{\sqrt{\sum_{k=1}^\infty Z_k^2}}\label{eq:SN_dist}
\end{align}
as $G\to \infty$,
where 
\begin{align*}
	p=\lim_{t\to \infty} \frac{P\left(S_g>t\right)}{P\left(\left|S_g\right|>t\right)},
\end{align*}
	$Z_k=(E_1+...+E_k)^{-1/\alpha}$ for each $k$, $\{E_k\}_k$ are i.i.d. standard exponential random variables, and $\{\epsilon_k\}_k$ are i.i.d. random variables that take the value of $1$ with probability $p$ and $-1$ with probability $(1-p)$ and are independent of $\{Z_k\}_k$.

}
We now claim that $SN_{2G}(\theta)$ converges in distribution to the same limiting distribution as \eqref{eq:SN_dist}. 
By Theorems 1 and 1$'$ in \cite{lepage1981convergence},
\begin{align}
&\left(\frac{1}{A_G}\sumg S_g,\frac{1}{A_G^2}\sumg S_g^2 \right)\nonumber\\
\stackrel{d}{\to }&(S,V):=\left(\sum_{k=1}^\infty\{\epsilon_k Z_k -(2p-1)\E[Z_k\1(Z_k < 1)]\},\sum_{k=1}^\infty Z_k^2\right)=O_p(1)\label{eq:rates_1}
\end{align}
holds for $A_G=G^{1/\alpha} L_1(G)$, where
$Z_k$, $\epsilon_k$, and $p$ are defined below Equation \eqref{eq:SN_dist}, and $L_1(\cdot)$ is slowly varying at $\infty$; and
\begin{align}
& \frac{1}{(A_G')^2}\sumg (X_g'X_g)^2\stackrel{d}{\to}\sum_{k=1}^\infty \tilde Z_k^2=O_p(1) \label{eq:rates_2}
\end{align}
holds
where $A_G'=G^{1/\alpha}L_2(G)$, $\tilde Z_k=(\tilde E_1+...+\tilde E_k)^{-1/\alpha}$ for each $k$, $\{\tilde E_k\}_k$ are i.i.d. standard exponential random variables, and $L_2(\cdot)$ is slowly varying at $\infty$.
Because $\alpha \in (1,2)$ and $L_1$ is slowly varying at $\infty$, Equation \eqref{eq:rates_1} implies the consistency
%
\begin{align}\label{eq:rate}
\|\hat \theta - \theta\|=\left\|\left(\sumg X_g'X_g\right)^{-1}\sumg S_g\right\|=O_p(L_1(G)G^{-(1-1/\alpha)})=o_p(1)
\end{align}
under Assumption \ref{assm}.
Using $
\hat U_g=U_g+X_g(\theta - \hat \theta)$ and
$\hat S_g =S_g+  X_g'X_g(\theta - \hat \theta)$, where $\hat U_g=(\hat U_{g1},...,\hat U_{ gN_g})'$,
 we can write
\begin{align*}
\frac{1}{A_G^{2}}\sumg\hat S_g^2 = & \frac{1}{A_G^{2}}\sumg S_g^2+\frac{1}{A_G^{2}}\sumg \left(\hat S_g - S_g \right)\hat S_g+ \frac{1}{A_G^{2}}\sumg S_g \left(\hat S_g -S_g\right)\\
=&\frac{1}{A_G^{2}}\sumg S_g^2+(1) + (2).
\end{align*}
We are going to show that the terms (1) and (2) are $o_p(1)$.
First,
\begin{align*}
\|(1)\|=&\left\|\frac{1}{A_G^{2}}\sumg (S_g+X_g'X_g(\theta  -\hat\theta))(X_g'X_g(\theta  -\hat\theta))'\right\|\\
\le& \left\|\frac{1}{A_G^{2}}\sumg S_g X_g'X_g\right\|\|\hat \theta - \theta\| + \left\|\frac{1}{A_G^{2}}\sumg (X_g'X_g)^2 \right\|\|\hat \theta - \theta\|^2\\
\le&
\underbrace{\sqrt{\frac{1}{A_G^{2}}\sumg S_g^2}}_{=O_p(1)} 
\underbrace{\sqrt{
	\frac{1}{A_G^{2}}\sumg (X_g'X_g)^2}}_{=O_p(1)}
\underbrace{ \|\hat \theta - \theta\| }_{=o_p(1)}
+
\underbrace{	\frac{1}{A_G^{2}}\sumg (X_g'X_g)^2}_{=O_p(1)}
	\underbrace{ \|\hat \theta - \theta\|^2 }_{=o_p(1)}
	\\
=&o_p(1)
\end{align*}
holds,
where the second inequality follows from the Cauchy-Schwarz inequality and the stochastic orders are due to Equations \eqref{eq:rates_1}, \eqref{eq:rates_2}, and \eqref{eq:rate}.  
Second, similar lines of calculations yield
\begin{align*}
\left\|(2)\right\|=&
\left\|\frac{1}{A_G^{2}}\sumg S_g(X_g'X_g ( \theta - \hat\theta))'\right\|	
=o_p(1).
\end{align*}
We have now established that
\begin{align*}
\frac{1}{A_G^{2}}\sumg\hat S_g^2=\frac{1}{A_G^{2}}\sumg S_g^2+o_p(1),
\end{align*}
and consequently, $SN_{1G}(\theta)$ is asymptotically equivalent to $ SN_{2G}(\theta)$. 
\bigskip 

\noindent \textbf{Step 2.}
We next show the validity of the CS bootstrap and subsampling. { Define the regular bootstrapped estimator
\begin{align*}
\check \theta_{b,j}=&\left(\sum_{g=1}^G w_g^jX_g'X_g\right)^{-1}\sum_{g\in \mathcal S_j}w_g^j X_g'Y_g.
\end{align*}
Since $B^{-1}-A^{-1}=A^{-1}(A-B)B^{-1}$, we have
\begin{align*}
\check \theta_{b,j}-\hat\theta_{b,j}=&\left(\sum_{g=1}^G w_g^j X_g'X_g\right)^{-1}\sum_{g=1}^G w_g^j X_g'Y_g-\left(\frac{G}{b}\right)  \left( \sum_{g=1}^G  X_g'X_g\right)^{-1}\sum_{g=1}^G w_g^j X_g' Y_g\\
=&\left(\frac{1}{G}\sumg X_g'X_g\right)^{-1}\left(\frac{1}{G}\sumg X_gX_g -\frac{1}{b}\sumg w_g^j X_g'X_g\right)\left(\frac{1}{b}\sumg w_g^j X_g'X_g\right)^{-1} \frac{1}{b} \sumg w_g^j X_g'Y_g\\
=& o_p(1) \cdot \check \theta_{b,j},
\end{align*}
where 
\[\left\|\frac{1}{G}\sumg X_gX_g -\frac{1}{b}\sumg w_g^j X_g'X_g\right\|=o_p(1)\]
follows from an application of the main theorem in \cite{csorgo1992law} under his Condition (1.5).
This implies
$\hat\theta_{b,j}=\check\theta_{b,j}(1+o_p(1))$. Therefore, in the bootstrap/subsampling process, $\check\theta_{b,j}$ can be replaced by $\hat\theta_{b,j}$ without changing the asymptotic behavior. Thus, it suffices to establish the validity of the CS bootstrap procedure based on the conventional $m$-out-of-$n$ estimator $\check\theta_{b,j}$. }

Now,
since the stable distributions $S$ and $V$ defined in the previous step are both continuous and $V>0$ with probability $1$, $S/V^{1/2}$ is continuously distributed and $J^*(\cdot)$ is continuous.
Hence, by invoking Theorem 4.1 in \cite{arcones1991additions}, we have
\begin{align*}
\sup_{t\in \Real}|\hat L_{G,b}(t)-J_G^*(t)|=o_p(1)
\end{align*}
as $M\to \infty$ and $G\to \infty$. Step 1, the triangle inequality, and the continuity of $J^*(\cdot)$   then imply
\begin{align*}
\sup_{t\in \Real}|\hat L_{G,b}(t)-J^*(t)|=o_p(1)
\end{align*}
This concludes the proof for the bootstrap.

The proof for subsampling proceeds in the same way above with with \cite{csorgo1992law} no longer needed and Theorem 4.1 in \cite{arcones1991additions} replaced by Theorem 11.3.1 in \cite{politis1999subsampling}.
\end{proof}

\subsection{Proof of Theorem \ref{cor:thin_tail}}\label{sec:proof:cor:thin_tail}

\begin{proof}[Proof of Theorem \ref{cor:thin_tail}]
We establish the result for both the CS bootstrap and subsampling.
The case of $\alpha<2$ follows directly from Lemma \ref{theorem:main}.	
For $\alpha=2$, the proof is similar to the proof of Lemma \ref{theorem:main} with the following minor modifications. 	
First, when $\alpha=2$, $S_g$ is in the domain of attraction of the normal distribution and hence Theorem 3.4 in \cite{gine1997student} yields 
	\begin{align*}
	SN_{1G}(\theta) \stackrel{d}{\to} \mathcal{N}(0,1).
	\end{align*}
	Second, to show the asymptotic equivalence of $SN_1(\theta)$ and $SN_2(\theta)$, note that both $S_g$ and $(X_g'X_g)$ belong to the domain of attraction of the normal law when $\alpha=2$.
	We branch into two cases.
	In case that both $S_g$ and $(X_g'X_g)$ have finite variances, we have
	$$\frac{1}{G} \sumg \hat S_g^2= \frac{1}{G} \sumg  S_g^2+o_p(1)\stackrel{p}{\to} \Var(S_g)$$ 
	by following the standard argument for consistency of the CR variance estimator.
	In case their variances do not exist, Lemma 3.1 in \cite{gine1997student} yields
	\begin{align*}
	\frac{1}{A_G^2}\sumg S_g^2 \stackrel{p}{\to} 1
	\end{align*}
	for $A_G$ such that
	\begin{align*}
	\frac{1}{A_G}\sumg (S_g - \E[S_g])\stackrel{d}{\to }\mathcal{N}(0,1).
	\end{align*}
	A similar argument holds when $S_g$ is replaced by $(X_g'X_g)$.
	Then, the arguments for bounding $\|(1)\|$ and $\|(2)\|$ in the proof of Lemma \ref{theorem:main} still go through, and thus for the self-normalized sums defined in Equation \eqref{eq:SN_sums}, it holds that $SN_2(\theta)=SN_1(\theta)+o_p(1)$.
	Finally, for the validity of the CS bootstrap, we now invoke Theorem 4.1 in \cite{arcones1991additions} for the bootstrap case and Theorem 2.2.1 in \cite{politis1999subsampling} for subsampling.
 Note that the limiting distribution is normal and hence continuous.
\end{proof}

\subsection{Proof of Theorem \ref{thm:iff}}\label{sec:proof:thm:iff}

\begin{proof}[Proof of Theorem \ref{thm:iff}]
The if part of the statement follows from the proof of Theorem \ref{cor:thin_tail}. The only if part is a direct implication of Theorem 3.4 in \cite{gine1997student} and the fact that for any $\alpha \in (1,2]$, the self-normalized sums defined in Equation \eqref{eq:SN_sums} satisfy $SN_2(\theta)=SN_1(\theta)+o_p(1)$, as shown in the proofs for Lemma \ref{theorem:main} and Theorem \ref{cor:thin_tail}.
	\end{proof}

\subsection{Proof of Theorem \ref{theorem:data-driven}}\label{sec:theorem:data-driven}
\begin{proof}[Proof of Theorem \ref{theorem:data-driven}]
Since we can condition on the realization of $b$ and take random subsets of clusters, it suffices to show that the data-driven choice $b$ as in Algorithm \ref{alg:bickel_sakov} satisfies that as $G\to\infty$, it holds with probability approaching one that 
\begin{align}
    \hat{b} \to \infty \text{ and } \hat{b}/G \to 0.  
\end{align}
It suffices to prove the statement for the univariate self-normalized sum 
\[
    T_G = \frac{\sum_g S_g}{\sqrt{\sum_g S_g^2}},
\]
as the general case follows under Assumption \ref{assm}. In particular, under Assumption~\ref{assm}, 
the design matrix is consistent for its limit
and, as shown in the proof of Lemma \ref{theorem:main}, replacing $S_g$ with $\hat S_g$ does not affect the asymptotic distribution.

Note that because the function $\phi$ used in the construction of $b_\ell$ is sublinear in Algorithm \ref{alg:bickel_sakov}, we have $\hat b/G \leq \lceil q \cdot \phi(G) \rceil / G \to 0.$
Then it remains to show $\hat b \to\infty$ with probability $1-o(1)$.

To this end, we follow the same argument as in the first half of the proof of Theorem~1 in \cite{bickel2008choice}, from the beginning up to their Equation~(22), by verifying Conditions (A.1)--(A.4) in that paper.  Also notice that excluding the case $b \sim G$, as in our Algorithm \ref{alg:bickel_sakov}, will not affect these parts of the argument.

    Condition (A.1) is immediate following the fact that with probability one, the mapping from data to $T_{G,b}$ is continuous.
    Conditions (A.2) and (A.3) are shown in the proof of our Lemma \ref{theorem:main}.

To check (A.4), for each $k=1,...,G$, define
\[
T_k =
\frac{\bar S_k}{\sqrt{\bar Q_k}}
=
\frac{\sum_{g=1}^{k} S_g}{\sqrt{\sum_{g=1}^{k} S_g^{2}}},
\qquad k\in\mathbb{N},
\]
and denote by \(L_k\) the CDF of \(T_k\) when \(S_1,\dots,S_k\stackrel{\text{i.i.d.}}{\sim}F\),
where \(F\) is a non-degenerate distribution in the domain of attraction of an \(\alpha\)-stable law for an $\alpha\in (1,2]$.
Condition~(A.4) requires the mapping \(k\mapsto L_k\) to be injective.
We now show that \(L_k\neq L_\ell\) whenever \(k\neq \ell\).
For any \(\mathbf{s}=(s_1,\dots,s_k)\in\mathbb{R}^k\),
Cauchy-Schwarz inequality implies that 
the support of \(L_k\) is contained in \([-\sqrt{k},\sqrt{k}]\). 
Let \(0<\varepsilon<\sqrt{k}\).
Because \(F\) is non‑degenerate, there exist \(a>0\) such that for any \(\delta \in (0,1)\), we have
\(
\bar p(\delta) :=P\bigl(S_g\in (a,(1+\delta)a)\bigr)>0.
\)
Consider the event
\[
E_k(\delta)
=
\{S_1,\dots,S_k\in(a,(1+\delta)a)\}.
\]
Observe that \(
P(E_k(a,\delta))=\bar p(\delta)^k>0.
\)
Furthermore, conditionally on \(E_k(a,\delta)\), we have
\(
\bar S_k\in(ka,k(1+\delta)a)
\)
and
\(
\bar Q_k\in(ka^{2},k(1+\delta)^2 a^{2}),
\)
so
\(
T_k
\in
(
\sqrt{k}/(1+\delta),
(1+\delta)\sqrt{k}
).
\)
Choosing $\varepsilon$ so that  \(\delta<\varepsilon/\sqrt{k}\) ensures
\(
T_k>\sqrt{k}-\varepsilon
\)
on \(E_k(\delta)\).
Hence,
\begin{equation*}
P\bigl(|T_k|>\sqrt{k}-\varepsilon\bigr)
\ge
\bar p(\delta)^k>0.
\end{equation*}
Since $\delta$ (and hence $\varepsilon$) can be made arbitrarily small, this implies that 
\(
\operatorname*{ess\,sup}|T_k|=\sqrt{k}.
\)
Take two integers \(k<\ell\) and define
\(
A_k=\{|t|>\sqrt{k}\}.
\)
Note that \(L_k(A_k)=0\).
Also, for \(\varepsilon=(\sqrt{\ell}-\sqrt{k})/2\),
\(
L_\ell(A_k)>0,
\)
so \(L_k\neq L_\ell\).
Since the argument holds for all such $k$ and $\ell$,
the map \(k\mapsto L_k\) is injective, showing that Condition~(A.4) is satisfied. 
\end{proof}

{

\subsection{Proof of Theorem \ref{theorem:bootstrap}}\label{sec:proof:theorem:bootstrap}

\begin{proof}[Proof of Theorem \ref{theorem:bootstrap}]
	Write 
	\begin{align*}
	T_{G}=&\frac{S_G}{\sqrt{V_G}}:=\frac{A_G^{-1}\sumg \left(\sumi Y_{gi}\right)}{\sqrt{A_G^{-2}\sumg \left(\sumi (Y_{gi}-\hat \theta)\right)^2}}
	\qquad\text{and}\\
	T_{G}^{*}=&\frac{S_G^*}{\sqrt{V_G^*}}:=\frac{A_G^{-1}\sumg v_g^*\left(\sumi Y_{gi}\right)}{\sqrt{A_G^{-2}\sumg \left(v_g^*\sumi (Y_{gi}-\hat\theta^*)\right)^2}}.
	\end{align*}
	Let $P$ denote the probability measure for the data and $P^*$ denote the probability measure of Rademacher auxiliary random variables.
	Define
	\begin{align*}
	p=\lim_{t\to \infty} \frac{P\left(\sumi Y_{gi}>t\right)}{P\left(\left|\sumi Y_{gi}\right|>t\right)}.
	\end{align*}
	Write $W_g=\left|\sumi Y_{gi}\right|$
	and the order statistics of $W_1,...,W_G$ as follows:
	\begin{align*}
	W_{G1}\ge& W_{G2}\ge ...\ge W_{GG}.
	\end{align*}
	The rescaled counterpart is denoted by $Z_{Gg}=A_G^{-1} W_{Gg}$, for  $g=1,...,G$ -- recall that $A_G=G^{1/\alpha} L(G)$ for a slow varying $L(\cdot)$  is defined right before Assumption \ref{assm}. For each $G$, we can collect them into a countably long vector
	\begin{align*}
	Z^G=(	Z_{G1},...,	Z_{GG},0,0,...)\in \Real^\infty.
	\end{align*} 
	Similarly defined is the countably long sign vector
	\begin{align*}
	\epsilon^G=(\epsilon_{G1},...,\epsilon_{GG},1,1,...)\in \Real^\infty,
	\end{align*}
	where $\epsilon_{Gg}$ indicates the sign such that $\sum_{i=1}^{N_h} Y_{hi}=\epsilon_{Gg}W_{Gg}$ for the cluster $h$ that corresponds to the $g$-th order statistic $W_{Gg}$ for each $g=1,...,G$, $G\in \mathbb N$. By Lemmas 1 and 2 in \cite{lepage1981convergence}, we have
	\begin{align*}
	Z^G \stackrel{d}{\to}& Z=(Z_1,Z_2,...)\quad\text{and}\quad
	\epsilon^G \stackrel{d}{\to} \epsilon=(\epsilon_1,\epsilon_2,...),
	\end{align*}
	where $\{Z_k\}_k$ and $\{\epsilon_k\}$ are defined in the proof for Lemma \ref{theorem:main}. 
	In addition, since $\Real^\infty$ is a complete separable metric space under the metric
	\begin{align*}
	d((x_1,x_2,...),(y_1,y_2,...))=\sum_{k=1}^\infty\frac{1}{2^k}\cdot \frac{|x_k-y_k|}{1+|x_k-y_k|} ,
	\end{align*}
	following Skorohod's representation theorem, on an adequately chosen probability space, 
	\begin{align*}
	d(Z^G,Z)\to 0 \quad\text{and}\quad d(\epsilon^G,\epsilon)\to 0
	\end{align*}
	$P$-almost surely.
	Denote the countable vector of i.i.d. Rademacher random variables by $v^*=(v_1^*,v_2^*,...)\in\Real^\infty$, which is invariant of $G$.
	We now claim that the weak convergence
	\begin{align*}
	S_G^*=\sumg \epsilon_{Gg} Z_{Gg} v_g^*\stackrel{d^*}{\to}S^*:=\sum_{k=1}^\infty \epsilon_k Z_k v_k^*
	\end{align*}
	for $(Z,\epsilon)$ with $P$-probability one,
	where the convergence in distribution $\stackrel{d^*}{\to}$ is with respect to $P^*$. Note that the limiting random variable on the right-hand side is well-defined since 
	\begin{align*}
	&\E^*\left[\epsilon_k Z_k v_k^*\right]=0 \text{ for all }k \text{ and } \\
	&\sum_{k=1}^\infty\E^*\left[(\epsilon_k Z_k v_k^*)^2\right]=\sum_{k=1}^\infty Z_k^2<\infty
	\end{align*}
	$P$-almost surely. 
	The convergence in distribution is shown following the same arguments as in the proof of Theorem 2 in \cite{knight1989bootstrap} with i.i.d. Rademacher random variables $v_k^*$ in place of their centered i.i.d. Poisson random variables $(M_k^*-1)$.
	Specifically, 
	observe that $Z_k\to 0$ as $k\to \infty$ $P$-almost surely.
	Following Equation (12) in the proof of Theorem 1 in \cite{lepage1981convergence}, define $\mathcal Z\subset \Real^\infty$ be the subspace consists of countable sequences $z=(z_1,z_2,...)$ such that $z_1\ge z_2\ge ...\ge 0$ (note that $\mathcal Z$ is also a complete separable space with the inherited topology). Subsequently, for a fixed $\varepsilon>0$,  define
	$\phi:\mathcal Z \times \{-1,1\}^\infty\times \{-1,1\}^\infty$ by
	\begin{align*}
	\phi(z,\epsilon,v^*)=
	\begin{cases}
	\sum_{k=1}^\infty \epsilon_k z_k \1(z_k >\epsilon) v_k^* \quad \text{ if } z_k\to 0 \text{ as } k\to \infty,\\
	0\qquad\qquad\qquad\qquad\qquad \text{ otherwise.}
	\end{cases}
	\end{align*}  
	Then $\phi$ is a continuous mapping with respect to the product topology. Thus by the continuous mapping theorem as well as the convergences of $d(Z^G,Z)\to 0$ and $d(\epsilon^G,\epsilon)\to 0$ with $P$-probability one established earlier,
	for any $\varepsilon>0$, 
	\begin{align*}
	\sum_{g=1}^G \epsilon_{Gg}Z_{Gg}\1(Z_{Gg}>\varepsilon )v_{g}^* \stackrel{d^*}{\to} \sum_{k=1}^\infty \epsilon_{k}Z_{k}\1(Z_{k}>\varepsilon )v_{k}^*
	\end{align*}
for $(Z,\epsilon)$ with $P$-probability one.
	In addition, note that
	\begin{align*}
	\E^*\left[\left( \sum_{g=1}^G \epsilon_{Gg}Z_{Gg}\1(Z_{Gg}\le\varepsilon )v_{g}^* \right)^2\right]= & \sum_{g=1}^G Z_{Gg}^2\1(Z_{Gg}\le\varepsilon )\Var^*(v_{k}^*)\le \sum_{k=1}^\infty Z_{k}^2\1(Z_{k}\le\varepsilon )
	\end{align*}
	holds almost surely in $P$ and the right-hand side converges to zero 
	as $\varepsilon\to 0$,
	which implies via Markov's inequality that, for any $\delta>0$,
	\begin{align*}
	\lim_{\varepsilon\to 0}\limsup_{G\to \infty}P^*\left( \left|\sum_{k=1}^\infty \epsilon_{Gk}Z_{Gk}\1(Z_{Gk}\le \varepsilon)v_k^*\right|> \delta\right)=0
	\end{align*}
	$P$-almost surely. Finally, for any $\delta>0$,
	\begin{align*}
	\lim_{\varepsilon\to 0}P^*\left( \left|\sum_{k=1}^\infty \epsilon_{k}Z_{k}\1(Z_{k}\le \varepsilon)v_k^*\right|> \delta\right)=0
	\end{align*}
	$P$-almost surely, which follows immediately from the fact that 
	\begin{align*}
	\E^*\left[ \left(\sum_{k=1}^\infty \epsilon_{k}Z_{k}\1(Z_{k}\le \varepsilon)v_k^*\right)^2 \right]=&\sum_{k=1}^\infty Z_{k}^2\1(Z_{k}\le \varepsilon)\to 0
	\end{align*}
	$P$-almost surely 
	as $\varepsilon\to 0$. Combining these results yields that
	\begin{align*}
	S_G^*\stackrel{d^*}{\to} S^* = \sum_{k=1}^\infty \epsilon_kZ_{k}v_k^*
	\end{align*}
	for $(Z,\epsilon)$ with $P$-probability one. 
	On the other hand, recall from Step 1 in the proof of Lemma \ref{theorem:main} that
	\begin{align*}
	S_G=\sumg \epsilon_{Gg}Z_{Gg}\stackrel{d}{\to}S:=\sum_{k=1}^\infty \{\epsilon_k Z_k-(2p-1)\E[Z_k\1(Z_k\le 1)]\},
	\end{align*}
	by Theorem 1 in \cite{lepage1981convergence}. Note that $Z_k$, $\epsilon_k$, and $v_k^*$ are all mutually independent from each other.
	Therefore, the  limiting distribution of $S_G^*$ given $Y_{1:G}$, i.e. $S^*$ conditionally on $(Z,\epsilon)$, differs from, $S$, the limiting $\alpha$-stable distribution of $S_G$  with positive $P$-probability.

	Next, to cope with the denominator term of $S_G^*$, note that, combined with the law of large numbers, the above weak convergence of $S_G^*$ also implies
	\begin{align*}
	\hat \theta^*=&\frac{1}{N}\sumg\epsilon_{Gg} W_{Gg} v_g^* \\
	=&\frac{1}{c+o_p(1)}\cdot \frac{1}{G}\sumg  \epsilon_{Gg} W_{Gg} v_g^*  \\
	=&\underbrace{\frac{1}{c+o_p(1)}}_{=O_p(1)}\cdot \underbrace{\frac{A_G}{G}}_{=\frac{L(G)}{G^{1-1/\alpha}}}\cdot \underbrace{\sumg  \epsilon_{Gg} Z_{Gg}  v_g^*}_{=O_p(1)}
	=o_p(1).
	\end{align*}
	Thus, the denominator term, $(V_G^*)^{1/2}$, of $S_G^*$ turns out to be asymptotically independent of the auxiliary Rademacher random variables  $v_g^*$:
	\begin{align*}
	V_G^*=\frac{1}{A_G^2}\sumg \left(v_g^*\sumi (Y_{gi}+o_p(1))\right)^2=	\sumg Z_{Gg}^2+o_p(1).
	\end{align*}
	Given $Y_{1:G}$, the denominator is asymptotically constant. Following Step 1 in the proof of Lemma \ref{theorem:main}, we have
	\begin{align*}
	V_G = \sumg Z_{Gg}^2 +o_p(1)\stackrel{d}{\to } \sum_{k=1}^\infty Z_k^2=O_p(1) .
	\end{align*}
	Thus, given $Y_{1:G}$, the denominator term  $(V_G^*)^{1/2}$ is a fixed value, while the original limit of the denominator is an $(\alpha/2)$-stable, non-degenerate continuous distribution. Hence, the limiting distribution of $V_G^*$ given $Y_{1:G}$ and the unconditional limiting distribution of $V_G$ differs with non-zero $P$-probability.

	Finally, note that $V_G^*>0$ $P$-almost surely. Thus, the fact that	
 $$(S_G^*,V_G^*)\stackrel{d^*}{\to} \left(\sum_{k=1}^\infty \epsilon_kZ_{k}v_k^*,\sum_{k=1}^\infty Z_k^2\right)$$ for almost every $(Z,\epsilon)$ and the continuous mapping theorem yield that
 \begin{align*}
 T_G^*\stackrel{d^*}{\to }\frac{\sum_{k=1}^\infty \epsilon_kZ_{k}v_k^*}{\sqrt{\sum_{k=1}^\infty Z_k^2}}
 \end{align*}
for $(Z,\epsilon)$ with $P$-probability one. This, together with the unconditional limiting distribution of $T_G$ implies 
  the conclusion that the unconditional limiting distribution of $T_G$ and the conditional limiting distribution of $T_G^*$ differs with positive $P$-probability. The inconsistency then follows.
\end{proof}

\subsection{Proof of Theorem \ref{theorem:uniformity}}\label{sec:proof:theorem:uniformity}
\begin{proof}[Proof of Theorem \ref{theorem:uniformity}]
	Let us first introduce the following lemma.
	\begin{lemma}[Weak convergence of triangular arrays]\label{lemma:weak_conv_varying_alpha_p}
		For any sequence of $P_G\in \mathbf P(\varepsilon)$ such that $\alpha_G\to \alpha_0\in [1+\varepsilon, 2]$ and $p_G\to p_0\in[0,1]$ as $G\to\infty$, we have
		\begin{align*}
			R_{1G}\stackrel{d}{\to }\mathbb S_{\alpha_0,p_0}.
		\end{align*} 
	\end{lemma}
	Its proof is presented in the end of this section.

Now, to show the statement of Theorem \ref{theorem:uniformity}, we shall derive the asymptotic distribution for the following self-normalized sums of $ S_g$:
\begin{align}
R_{1G}:=\frac{\sum_{g=1}^G  S_g}{\sqrt{\sum_{g=1}^G  S_g^2}}\quad\text{ and }\quad
R_{2G}:=\frac{\hat \delta - \delta}{\hat \sigma}=\frac{\sum_{g=1}^G  S_g}{\sqrt{\sum_{g=1}^G  \hat S_g^2}}.
\end{align}
Following Eq (1.3) in \cite{logan1973limit}, we obtain
\begin{align*}
R_{2G}=	R_{1G}\left(\frac{G}{G-	R_{1G}^2}\right)^{1/2}.
\end{align*}
Thus, by Lemma \ref{lemma:weak_conv_varying_alpha_p}, the limiting distribution of $R_{2G}$ coincides with the one of $R_{1G}$.


The proof follows a similar structure to the one for Theorem 3.1 in \cite{romano2012uniform}.
We will apply our Lemma \ref{lemma:high-level_uniformity} in Appendix \ref{sec:auxiliary_Lemmas} with 
\begin{align*}
	R_G=&\frac{\hat \delta - \delta}{\hat \sigma} \quad\text{ and }\quad
	\hat R_b=\frac{\hat \delta_{b,j} - \hat\delta}{\hat \sigma_{b,j}}.
\end{align*}
First, we verify
\begin{align}
	\sup_{P\in\mathbf P}\sup_{x\in\mathbb R}|J_b(x,P)-J_G(x,P)|\to 0\label{eq:unif_CDF_conv}
\end{align}
as $b,G\to \infty$ with $b/G=o(1)$.
By way of contradiction, assume that it fails. Then, there exists a subsequence $G_l$ and some $(\alpha,p)\in [1+\varepsilon,2]\times [0,1]$ such that either
\begin{align*}
	\sup_{x\in\mathbb R}|J_{b_{G_l}}(x,P_{G_l})- F_{\alpha,p}(x)|\not\to 0\quad
	\text{or}\quad	\sup_{x\in\mathbb R}|J_{G_l}(x,P_{G_l})-F_{\alpha,p}(x)|\not\to 0.
\end{align*}
Recall that $\mathbb S_{\alpha,p}\sim F_{\alpha,p}$ has a continuous distribution (almost everywhere). 
Yet, either of these would violate Lemma \ref{lemma:weak_conv_varying_alpha_p}. Thus Condition (\ref{eq:unif_CDF_conv}) must hold.

We will next verify
the condition that
\begin{align*}
	\sup_{P\in \mathbf P}P\left(\sup_{x\in \Real}\left|\hat L_G(x)-L_G(x,P)\right|>\varepsilon'\right)=o(1)
\end{align*}
for all $\varepsilon'>0$.
Consider any sequence $\{P_G\in \mathbf P:G \ge 1\}$. For any $\eta>0$, we have
\begin{align*}
	&\sup_{x\in\mathbb R}\{\hat L_G(x)-L_G(x,P_G)\}\\
	\le &
	\sup_{x\in\mathbb R}\{\hat L_G(x)-L_G(x+\eta,P_G)\}+
	\sup_{x\in\mathbb R}\{ L_G(x+\eta,P_G)-L_G(x,P_G)\}\\
	\le&
	\sup_{x\in\mathbb R}\{\hat L_G(x)-L_G(x+\eta,P_G)\}+
	\sup_{x\in\mathbb R}\{ L_G(x+\eta,P_G)-L_b(x+\eta,P_G)\}\\
	+&\sup_{x\in\mathbb R}\{ L_b(x,P_G)-L_G(x,P_G)\}
	+
	\sup_{x\in\mathbb R}\{ L_b(x+\eta,P_G)-L_b(x,P_G)\}\\
	=&(i)+(ii)+(iii)+(iv).
\end{align*}
Note that $(ii)$ and $(iii)$ are both $o_{P_G}(1)$ by Lemma 4.5 in \cite{romano2012uniform}. Furthermore, $(iv)$ converges to zero as $\eta\to 0$. 

Finally, we will verify $(i)= o_{P_G}(1)$ as $\eta\to 0$.  By considering a subsequence, if necessary, one may assume without loss of generality that $P_G$ is such that $\alpha_G\to \alpha$ and $p_G\to p$. The proof for this statement utilizes an argument similar to those taken in Theorem 11.3.1 in \cite{politis1999subsampling}. By its definition,
\begin{align*}
	\hat L_G(x)=&\frac{1}{B_G}\sum_{j=1}^{B_G}\1\left\{\frac{\hat \delta_{b,j} - \hat \delta}{\hat \sigma_{b,j}}\le x\right\}\\
	\le&\frac{1}{B_G}\sum_{j=1}^{B_G}\1\left\{\frac{\hat \delta_{b,j} -  \delta}{\hat \sigma_{b,j}}\le x+ \frac{\hat \delta - \delta}{\hat \sigma_{b,j}}\right\}\\
	\le&
	\frac{1}{B_G}\sum_{j=1}^{B_G}\1\left\{\frac{\hat \delta_{b,j} -  \delta}{\hat \sigma_{b,j}}\le x+\eta\right\} + (1-R_G(\eta)),
\end{align*}
where $R_G(\eta)$ is defined for $\eta>0$ as
\begin{align*}
	R_G(\eta)=&\frac{1}{B_G}\sum_{j=1}^{B_G}\1\left\{\frac{\hat \delta - \delta}{\hat \sigma_{b,j}}\le \eta\right\}\\
	=&\frac{1}{B_G}\sum_{j=1}^{B_G}\1\left\{ 
	(b/A_b) \hat \sigma_{b,j}\ge (b/A_b)  (\hat \delta -\delta)/\eta
	\right\},
\end{align*}
$A_b=b^{1/\alpha}L(b)$ for some slow varying $L$ at infinity. As $A_G/A_b\to 0$, for any $\varepsilon''>0$, it holds that $(b/A_b)(\hat \delta -\delta)\le \varepsilon''$ with probability approaching one along $P_G$. This is because $\hat \delta$ is the full sample estimator and thus $(G/A_G)(\hat\delta - \delta)=O_{P_G}(1)$ follows from the proof of Lemma \ref{lemma:weak_conv_varying_alpha_p}. As such, following the proof of Lemma \ref{lemma:weak_conv_varying_alpha_p}, we have
\begin{align*}
	R_G(\eta) \ge \frac{1}{B_G}\sum_{j=1}^{B_G} \1\left\{(b/A_b)\hat \sigma_{b,j}\ge \varepsilon''/\eta\right\}\stackrel{P_G}{\to} P_G(V \ge\varepsilon''/\eta)
\end{align*}
as $G\to \infty$,
where $V$ is the stable distribution with index of stability of $\alpha/2$. By Theorem 1$'$ in \cite{lepage1981convergence} for example, $V$ has the representation
$
V=\sum_{k=1}^\infty Z_k^2,
$ where
$Z_k=(E_1+...+E_k)^{-1/\alpha}$ for each $k$, $\{E_k\}_k$ are i.i.d. standard exponential random variables.
As $\varepsilon''$ can be arbitrarily small, we have $R_G(\eta)= 1 +o_{P_G}(1)$. Thus, we have
\begin{align*}
	\hat L_G(x)
	\le&
	\frac{1}{B_G}\sum_{j=1}^{B_G}\1\left\{\frac{\hat \delta_{b,j} -  \delta}{\hat \sigma_{b,j}}\le x+\eta\right\} + (1-R_G(\eta))\\
	\le&
	L_G(x+\eta,P_G)+o_{P_G}(1).
\end{align*}
A similar argument derives $\hat L_G(x) \ge L_G(x+\eta,P_G)+o_{P_G}(1)$. This shows $(i)=o_{P_G}(1)$ as $\eta\to 0$, and hence concludes the proof of Theorem \ref{theorem:uniformity} for the case of subsampling. 

Finally, under the regime of $b^2/G\to 0$, the result follows from the asymptotic equivalence between subsampling and the $m$-out-of-$n$ bootstrap under the regime $m^{2}/n \to 0$; see Corollary 2.3.1 in \cite{politis1999subsampling}.
\end{proof}
\begin{proof}[Proof of Lemma \ref{lemma:weak_conv_varying_alpha_p}]
	First, consider the case of $ \alpha_0<2$.
	Denote $S_g=S_g(\alpha,p)$ to emphasize the dependence of the DGP on the index $\alpha$ of stability and the tail balancing parameter $p$. (It does not suggest that the DGP is uniquely defined by these two parameters.) For each DGP, $P_G \in \{P_G:G\ge 1\}\subset \mathbf P_1(\varepsilon)$, with indices $(\alpha_m,p_m)$ for an auxiliary index $m=G$,
	define
	\begin{align*}
	X_{mn}=\frac{\sum_{g=1}^n S_g(\alpha_m,p_m)}{\sqrt{\sum_{g=1}^n S_g^2(\alpha_m,p_m)}}
	\end{align*}
	for each $n\ge 1$.
	Since $(\alpha_m,p_m)$ is fixed over $n$ for each $m$, we can apply Theorem 1$'$ in \cite{lepage1981convergence} to obtain that, for each $m$ as $n\to \infty$, there exists some positive sequence $A_{mn}\to \infty$ such that
	\begin{align*}
	&\left(\frac{1}{A_{mn}}\sum_{g=1}^n S_g(\alpha_m,p_m),\frac{1}{A_{mn}^2}\sum_{g=1}^n S_g^2(\alpha_m,p_m) \right)\\
	&\stackrel{d}{\to }\left(\sum_{k=1}^\infty\{\epsilon_k(p_m) Z_k(\alpha_m) -(2p_m-1)\E[Z_k(\alpha_m)\1(Z_k(\alpha_m) < 1)]\},\sum_{k=1}^\infty Z_k^2(\alpha_m)\right)=(S_m,V_m)
	\end{align*}
	as $n\to\infty$,
	where  $Z_k(\alpha_m)=(E_1+...+E_k)^{-1/\alpha_m}$ for each $k$, $\{E_k\}_k$ are i.i.d. standard exponential random variables, and $\{\epsilon_k(p_m)\}_k$ are i.i.d. random variables that take the value of $1$ with probability $p_m$ and $-1$ with probability $(1-p_m)$ and are independent of $\{Z_k(\alpha_m)\}_k$. Note that the distributions of both $S_m$ and $V_m$ are stable with indices of stability of $\alpha_m$ and $\alpha_m/2$, respectively.
	Furthermore, it follows from Corollary 1 in \cite{lepage1981convergence} that
	\begin{align*}
	&X_{mn}\stackrel{d}{\to} X_m\stackrel{d}{=}\frac{\sum_{k=1}^\infty\{\epsilon_k(p_m) Z_k(\alpha_m) -(2p_m-1)\E[Z_k(\alpha_m)\1(Z_k(\alpha_m) < 1)]\}}{\sqrt{\sum_{k=1}^\infty Z_k^2(\alpha_m)}}.
	\end{align*}
	Let the limiting distribution on the right-hand side be denoted by $\mathbb S_{\alpha_m,p_m}$.
	Also, note that $(\alpha_m,p_m)\to (\alpha_0,p_0)$ by our construction, and thus, $$X_m\stackrel{d}{\to} X\sim \mathbb S_{\alpha_0,p_0}$$ 
	follows from the convergence of the sequence of the characteristic functions of the stable distributions $S_m$ and $V_m$, as these characteristic functions are continuous in $(\alpha,p)$ over $(1,2)\times [0,1]$ (cf. Remark 4 on page 7 in \citealp{samorodnitsky1994stable}) and $V_m$ is positive with probability one for all $\alpha\in (1,2)$. 
	
	Next, by invoking the Skorohod's representation theorem (as $\mathbb R$ is a separable metric space), there exist versions of $X_{mn}$ and $X_m$ such that $X_{mn}\stackrel{a.s.}{\to} X_m$ for each $m$ and as $n\to \infty$, and $X_{m}\stackrel{a.s.}{\to} X$ as $m\to \infty$.
	Now, for such $X_{mn}$, define $Y_n=X_{nn}$. 	By construction, we have $Y_n\overset{d}{=}R_{1n}$ for all $n\ge 1$. 
	Also, it follows from the almost sure converges that 
	\begin{align*}
	\lim_{M\to\infty}\limsup_{n\to \infty} P(|X_{mn}-Y_n|\ge \varepsilon)=0
	\end{align*}
	for all $\varepsilon>0$.
	Applying Lemma \ref{lemma:second_conv_together} in Appendix \ref{sec:auxiliary_Lemmas}, we have
	$
	Y_n\stackrel{d}{\to }X
	$
	as $n\to \infty$.
	Thus, we conclude
	$
	R_{1n}\stackrel{d}{\to }X.
	$

	Now, consider the case of $\alpha_0=2$. We only need to consider the case where we have $\alpha_G<2$ for at least one $G$, as, otherwise, $\alpha_G=2$ for all $G$ and 
	\begin{align*}
	R_{1G}\stackrel{d}{\to } \mathcal{N}(0,1)
	\end{align*}
	follows immediately from the Lindeberg-Feller CLT. Now, for those $\alpha_m<2$, construct $ X_{mn}$ as in the previous case. By Corollary in \cite{lepage1981convergence}, we have
	\begin{align*}
	X_{mn}=\frac{\sum_{g=1}^n S_g(\alpha_m,p_m)}{\sqrt{\sum_{g=1}^n S_g^2(\alpha_m,p_m)}}\stackrel{d}{\to} X_m\sim \mathbb S_{\alpha_m,p_m}.
	\end{align*}
	By Assertion (vi) in Section 5 and Equation (5.13) in \cite{logan1973limit}, the density $f_{\alpha_m,p_m}(\cdot)$ of $\mathbb S_{\alpha_m,p_m}$ exists and is bounded everywhere except on a set with measure zero, and, as $\alpha_m\to 2$, $f_{\alpha_m,p_m}\to \varphi$, the standard normal density, on the real line. Thus, by the bounded convergence theorem, the CDF $F_{\alpha_m,p_m}(x)$ of $\mathbb S_{\alpha_m,p_m}$ converges to the standard normal distribution function $\Phi(x)$ for all $x\in \Real$, i.e. $X_m\stackrel{d}{\to}X \sim \mathcal{N}(0,1)$. Using the same construction of $Y_n$ as above, we conclude $R_{1n}\stackrel{d}{\to } \mathcal{N}(0,1)$ by Lemma \ref{lemma:second_conv_together} in Appendix \ref{sec:auxiliary_Lemmas}.
\end{proof}


\section{Auxiliary Lemmas from the Literature}\label{sec:auxiliary_Lemmas}

The following lemma restates Theorems 2.1 and 2.2 as well as Remark 2.1 in \cite{romano2012uniform} for convenience of reference.
\begin{lemma}[High-level uniformity]\label{lemma:high-level_uniformity}
For subsampling and 
	under the setup in Section \ref{sec:theory:uniform},
	\begin{align*}
	\lim_{G\to \infty}\sup_{P\in\mathbf P}\sup_{x\in \Real}|J_b(x,P)-J_G(x,P)|=0,
	\end{align*}
	implies
	\begin{align*}
	\liminf_{G\to \infty}\inf_{P\in \mathbf P}P\left(L_G^{-1}(a_1,P)\le R_G \le  L_G^{-1}(1-a_2,P)\right)\ge 1-a_1-a_2
	\end{align*}
	for any nonnegative $a_1$ and $a_2$ such that $0\le a_1+a_2<1$. In addition, if $J_G(x,P)$
	tends in distribution to a limiting distribution $J(x,P)$ that is continuous, then
	\begin{align*}
	\lim_{G\to \infty}\inf_{P\in \mathbf P}P\left(L_G^{-1}(a_1,P)\le R_G \le  L_G^{-1}(1-a_2,P)\right)= 1-a_1-a_2.
	\end{align*}
	Finally, if 
	\begin{align*}
	\sup_{P\in \mathbf P}P\left(\sup_{x\in \Real}\left|\hat L_G(x)-L_G(x,P)\right|>\varepsilon\right)=o(1)
	\end{align*}
	for all $\varepsilon>0$, then
	\begin{align*}
	\lim_{G\to \infty}\inf_{P\in \mathbf P}P\left(\hat L_G^{-1}(a_1)\le  R_G \le  \hat L_G^{-1}(1-a_2)\right)= 1-a_1-a_2.
	\end{align*}
\end{lemma}

\medskip
The next result is identical to Theorem 3.5 in \cite{resnick2008extreme}.
\begin{lemma}[Second converging together theorem]\label{lemma:second_conv_together}
	Suppose that $\{X_{mn},X_m,X,Y_n: n\ge 1, m\ge 1\}$ are random elements of the metric space $(\mathbb S,\mathcal S)$  with a metric $d(\cdot,\cdot)$ that are defined on a common domain. Assume that for each $m$, as $n\to \infty$,
	$
	X_{mn}\leadsto X_m,
	$
	and as $m\to \infty$,
	$
	X_m\leadsto X,
	$
	Further suppose that for all $\varepsilon>0$,
	\begin{align*}
	\lim_{m\to\infty}\limsup_{n\to \infty} P(d(X_{mn},Y_n)\ge \varepsilon)=0.
	\end{align*}
	Then, as $n\to \infty$, we have
	$
	Y_n\leadsto X,
	$
	where $\leadsto$ denotes weak convergence. 
\end{lemma}
}
\newpage
\bibliographystyle{ecta}
\bibliography{biblio}
\end{document}